\pdfoutput=1
\documentclass[12pt,a4paper]{amsart}
\usepackage[a4paper,inner=2.5cm,outer=2.5cm,top=2.5cm,bottom=2.5cm]{geometry}
\usepackage{amsmath,amssymb,amsthm,enumerate,mathtools,stmaryrd
}
\usepackage{hyperref}
\hypersetup{colorlinks=true,linkcolor=blue,citecolor=teal,filecolor=magenta,urlcolor=cyan}
\usepackage{longtable}

\usepackage{makecell}
\usepackage{graphicx}

\usepackage{float}

\usepackage[all]{xy}
\usepackage{extarrows}

\usepackage{xcolor}

\def\nonregular{{special}}

\usepackage{url}

\usepackage[backend=bibtex,style=alphabetic,sorting=nyt,isbn=false,url=false,doi=true,maxalphanames=10,minalphanames=4,mincitenames=4,maxcitenames=10,minnames=4,maxnames=10,giveninits=true,maxbibnames=99]{biblatex}
\theoremstyle{plain}
\newtheorem{theorem}{Theorem}[section]
\newtheorem{proposition}[theorem]{Proposition}

\newtheorem{lemma}[theorem]{Lemma}

\theoremstyle{definition}
\newtheorem{definition}[theorem]{Definition}
\newtheorem{example}[theorem]{Example}
\newtheorem{notation}[theorem]{Notation}

\makeatletter
\@addtoreset{proofpart}{theorem}
\makeatother
\theoremstyle{remark}
\newtheorem{remark}[theorem]{Remark}

\def\nonregular{{special}}

\newcommand{\Z}{\mathbb{Z}}
\newcommand{\C}{\mathbb{C}}

\newcommand{\cS}{\mathcal{S}}
\newcommand{\cA}{\mathcal{A}}

\newcommand{\cP}{\mathcal{P}}

\newcommand{\cW}{\mathcal{W}}
\newcommand{\cT}{\mathcal{T}}
\DeclareMathOperator{\Aut}{Aut}
\newcommand{\VEV}[1]{{\big\langle 0 \big| {#1} \big| 0 \big\rangle}}
\newcommand{\VEVc}[1]{{\big\langle 0 \big| {#1} \big| 0 \big\rangle^\circ}}

\newcommand{\res}{\mathop{\rm res}}

\newcommand{\vev}[1]{\langle 0 | {#1} | 0 \rangle}

\newcommand{\restr}[2]{\mathop{\big\lfloor_{{#1}\to {#2}}}}
\newcommand{\set}[1]{\llbracket {#1} \rrbracket}
\newcommand{\bt}{\mathbf{t}}
\newcommand{\bs}{\mathbf{s}}

\newcommand{\ii}{\mathrm{i}}
\newcommand{\np}{\mathsf{np}}

\newcommand{\fS}{\mathfrak{S}}

\title[Blobbed topological recursion and KP integrability]
{Blobbed topological recursion and KP integrability}

\author[A.~Alexandrov]{A.~Alexandrov}
\address{A.~A.: Center for Geometry and Physics, Institute for Basic Science (IBS), Pohang 37673, Korea
}
\email{alex@ibs.re.kr}

\author[B.~Bychkov]{B.~Bychkov}
\address{B.~B.: Department of Mathematics, University of Haifa, Mount Carmel, 3498838, Haifa, Israel}
\email{bbychkov@hse.ru}

\author[P.~Dunin-Barkowski]{P.~Dunin-Barkowski}
\address{P.~D.-B.: Faculty of Mathematics, HSE University, Usacheva 6, 119048 Moscow, Russia; HSE--Skoltech International Laboratory of Representation Theory and Mathematical Physics, Skoltech, Bolshoy Boulevard 30 bld. 1, 121205 Moscow, Russia; and NRC “Kurchatov Institute” -- ITEP, 117218 Moscow, Russia}
\email{ptdunin@hse.ru}

\author[M.~Kazarian]{M.~Kazarian}
\address{M.~K.: Faculty of Mathematics, HSE University, Usacheva 6, 119048 Moscow, Russia; and Igor Krichever Center for Advanced Studies, Skoltech, Bolshoy Boulevard 30 bld. 1, 121205 Moscow, Russia}
\email{kazarian@mccme.ru}

\author[S.~Shadrin]{S.~Shadrin}
\address{S.~S.: Korteweg-de Vries Institute for Mathematics, University of Amsterdam, Postbus 94248, 1090GE Amsterdam, The Netherlands}
\email{S.Shadrin@uva.nl}

\begin{document}
	
\begin{abstract}
		We revise the notion of the blobbed topological recursion by extending it to the setting of generalized topological recursion as well as allowing blobs which do not necessarily admit topological expansion. We show that the so-called non-perturbative differentials form a special case of this revisited version of blobbed topological recursion. Furthermore, we prove the KP integrability of the differentials of blobbed topological recursion for the input data that include KP-integrable blobs. This result generalizes, unifies, and gives a new proof of the KP integrability of nonperturbative differentials conjectured by Borot--Eynard and recently proved by the authors.
\end{abstract}

\maketitle
	
\tableofcontents
	
\section{Introduction}

In this paper we deal with the following concepts:

\begin{itemize}
	\item Topological recursion of Chekhov--Eynard--Orantin (the CEO topological recursion)~\cite{EO-1st} is a way to uniformly solve a huge variety of enumerative problems recursively. Its wide range of connections to the intersection theory of moduli spaces, matrix models, and integrability brings new insight in every instance of its applications.
	\item Blobbed topological recursion~\cite{BS-blobbed} is an interface to encode more general solutions of the so-called abstract loop equations. It substantially extends the range of applications of topological recursion, see e.g.~\cite{AS-blobbed,BHW-blobbed,bonzom2020blobbedtopologicalrecursioncorrelation}.
	\item Generalized topological recursion~\cite{alexandrov2024degenerateirregulartopologicalrecursion} is a very recently defined generalization of the original CEO topological recursion that removes all restrictions on the input. It provides interesting and meaningful answers in very unexpected situations~\cite{CGS}.
	\item KP integrability is a basic integrability property that keeps to reoccur in very different areas of mathematics. It can be understood globally as a property of a system of differentials, and it has close ties with topological recursion~\cite{BorEyn-AllOrderConjecture,eynard2024hirotafaygeometry,alexandrov2024topologicalrecursionrationalspectral,ABDKS3,ABDKSnp}.
\end{itemize}

Due to an enormous amount of applications of these concepts, it is very interesting to settle the foundational questions about their interrelation and possible context for the ``ultimate'' generalizations of the basic definitions. The main goal of the paper is to provide a new context to an old conjecture of Borot and Eynard~\cite{BorEyn-AllOrderConjecture,eynard2024hirotafaygeometry}, which states that non-perturbative differentials of topological recursion are KP integrable. This conjecture was recently proved~\cite{ABDKSnp}, and the implications of this result, for instance, for the knot theory, are still to be explored.

In a nutshell, the differentials of topological recursion defined on the Cartesian powers of a given Riemann surface $\Sigma$ are merged via some convolution procedure with the so-called Krichever differentials. The latter ones are also  differentials on the Cartesian powers of $\Sigma$ and their expansions give algebro-geometric solutions of the KP hierarchy~\cite{Krichever-main,KricheverShiota}, while the differentials themselves provide a setup for conformal field theory on $\Sigma$~\cite{Kawamoto,Ooguri}. The resulting convolution differentials are the output of non-perturbative topological recursion~\cite{EynardMarino,BorEyn-AllOrderConjecture,BorEyn-knots}.

The structure of the convolution procedure is very similar to the interface proposed by the blobbed topological recursion. And indeed, we show that the non-perturbative differentials give an instance of an extended version of blobbed topological recursion, where we allow to have non-topological expansions for the blobs. The blobs in this case are the Krichever differentials.

This, in fact, prompts also a revision of the definition of the blobbed topological recursion. We propose a new construction of generalized topological recursion with blobs, where the blobs essentially replace the Bergman kernel data. This provides us a very flexible deformation toolkit given by the freedom of choice of blobs. As a result, we give a new proof of the Borot--Eynard conjecture, conceptually different from the original one presented in~\cite{ABDKSnp}.

We also take a closer look at the KP integrability property for a system of differentials~\cite{ABDKS3}. We explain how our global viewpoint matches existing constructions of algebro-geometric solutions of KP and multi-KP hierarchy~\cite{Krichever-main,krichever2023quasiperiodicsolutionsuniversalhierarchy}. Interesting effect occurs  in the analysis of the convolution procedure that is a cornerstone of the constructions of the blobbed topological recursion and non-perturbative differentials. Namely, we prove that the convolution of KP-integrable systems of differentials produces a KP integrable system of differentials.

In our review of KP integrability considered as a global property for systems of differentials, as in~\cite{ABDKS3}, we also make a special accent on the multi-KP integrability aspects~\cite{Kac-vdLeur} that occur once we expand a given system of differentials at several points on the curve. Typically, it is done at the poles of $x$, one of the functions that forms the input data of topological recursion (see an extended explanation in~\cite{eynard2024hirotafaygeometry}). However, one can choose any number of arbitrary points and we discuss the corresponding statement in detail.

It worth to mention that even if the underlying Riemann surface $\Sigma$ is of genus $0$, and thus there are no non-perturbative effects and the (generalized) topological recursion is KP integrable by itself~\cite{alexandrov2024topologicalrecursionrationalspectral}, our discussion of multi-KP integrability is very useful. Namely, topological recursion is known to interact with the genus expansion of Dubrovin--Frobenius manifolds defined via Dubrovin's superpotential~\cite{dunin-super-1,dunin-super-2}. As a corollary, our discussion of multi-KP integrability conceptually explains the effects observed in the study of the integrable systems of topological type for the Hurwitz--Frobenius manifolds in genus zero~\cite{Dub-2dTFT}, where the multi-component KP system appears as a reduction of Givental--Milanov type approach to the Dubrovin--Zhang hierarchies, see a forthcoming paper~\cite{Carlet-vdLeur-Shadrin}.

\subsection{Notation} Throughout the paper we use the following notation:

\begin{itemize}
	\item
Let $\set n$ be the set of indices $\{1,\dots,n\}$. For each $I\subset \set n$ let $z_I$ denote the set of formal variables / points on a Riemann surface / functions (respectively, depending on the context) indexed by $i\in I$.

\item

For the $n$-th symmetric group $\fS_n$ its subset of cycles of length $n$ is denoted by $C_n$. Let $\det^\circ (A_{ij})$ denote the connected determinant of a matrix $(A_{ij})$ given in the case of an $n\times n$ matrix by
\begin{align}
	{\det}^\circ (A_{ij})  \coloneqq (-1)^{n-1} \sum_{\sigma \in C_n} A_{i,\sigma(i)}.
\end{align}

\item
The operator $\restr{z}{z'}$ is the operator of restriction of an argument $z$ to the value $z'$, that is, $\restr{z}{z'}f(z) = f(z')$. When we write $\restr{}{z}$ we mean that the argument of a function or a differential to which this operator applies is set to $z$, without specifying the notation for the former argument.


\item
Let $\cS(z)$ denote $z^{-1}(e^{\frac z2}-e^{-\frac z2} )$.

\end{itemize}

\subsection{Organization of the paper} We keep this paper focused on conceptual revision of the definitions and their interrelations and do not include any applications or potential applications. It is a conscious choice, as there is enormous literature on applications in any case, while the foundational questions are often underrepresented.

In Section~\ref{sec:KPintegrability} we give a survey of KP integrability, in particular in the global context as a property of a system of multi-differentials on an algebraic curve.

In Section~\ref{sec:Convolution} we discuss the main technical tool that we apply to systems of differential, the so-called \emph{convolution}, without any reference to their further application for topological recursion. 

Finally, in Section~\ref{sec:Blobbed} we use this convolution to discuss the extended version of blobbed topological recursion and its relation to KP integrability and its relation to the non-perturbative topological recursion.

\subsection{Acknowledgments}
The authors thank KoBus, the Korean bus company that serves the line between Pohang and the Incheon International Airport, for the excellent working conditions and stimulating environment that greatly contributed to the success of this work.
	
	A.~A. was supported by the Institute for Basic Science (IBS-R003-D1). A.~A. is grateful to BIMSA in Beijing, USTC in Hefei, and HSE in Moscow for hospitality.
	B.~B. was supported by the ISF Grant 876/20. P.D.-B. and M.K. were supported by the Basic Research Program of the National Research University Higher School of Economics. S.~S. was supported by the Dutch Research Council (OCENW.M.21.233). 
	
	\section{Review of global KP integrability}
	
	\label{sec:KPintegrability}
	
	\subsection{KP integrability for a system of differentials}
	
	Let $\Sigma$ be a possibly non-compact complex curve and $\{\omega_n\}$, $n\ge1$ be a system of symmetric meromorphic $n$-differentials on $\Sigma$.
	
	\begin{definition}
		We say that the system $\{\omega_n\}$ is KP integrable if there exists a bi-halfdifferential $K$ on~$\Sigma$ with a singularity on the diagonal and the local behavior
		\begin{equation}
			K(p_1,p_2)=\left(\frac1{z_1-z_2}+(\text{holomorphic})\right) \sqrt{dz_1}\sqrt{dz_2},\quad z_i=z(p_i),
		\end{equation}
		such that the following determinantal equalities hold
		\begin{align}
			\label{eq:omega1K}\omega_1(p_1)&=\restr{p_2}{p_1}\left(K(p_1,p_2)-\frac{\sqrt{dz_1}\sqrt{z_2}}{z_1-z_2}\right),\quad z_i=z(p_i),\\
			\label{eq:omeganK}\omega_n(p_{\set n})&=\det\nolimits^\circ\|K(p_i,p_j)\|_{i,j=1,\dots,n}.
		\end{align}
		Here $z$ is some local coordinate on~$\Sigma$ or just a meromorphic function. Neither the condition on the behavior of~$K$ near the diagonal nor an expression for $\omega_1$ depend on a choice of~$z$.
		
		We assume that either $K$ is meromorphic or $K$ is a formal series in an additional formal parameter~$\hbar$ such that the coefficient of any power of $\hbar$ is meromorphic. Then, the same holds for~$\omega_n$: it is global meromorphic or it is a series in~$\hbar$ with meromorphic coefficients, respectively.
	\end{definition}
	
	The term ``halfdifferential" means that $K$ is a section of a suitable square root of the canonical bundle over $\Sigma$ with respect to each of its arguments.
	
	The differential $\omega_2(p_1,p_2)=-K(p_1,p_2)K(p_2,p_1)$ has a pole of order two on the diagonal, and one can show that $\omega_n-\delta_{n,2}\frac{dz_1dz_2}{(z_1-z_2)^2}$ is regular on the diagonals.

	In fact, the kernel~$K$ can be expressed in terms of the differentials $\omega_n$ by the following explicit formula
	\begin{equation}
		K(p_1,p_2)=\frac{\sqrt{dz_1}\sqrt{dz_2}}{z_1-z_2}\exp\left(\sum_{n=1}^\infty\frac{1}{n!}
		\biggl(\int\limits_{p_2}^{p_1}\biggr)^n\Bigl(\omega_n(\tilde p_{\set{n}})
		-\delta_{n,2}\tfrac{d\tilde z_1d\tilde z_2}{(\tilde z_1-\tilde z_2)^2}\Bigr)\right)
	\end{equation}
	that holds for any choice of the local coordinate~$z$. To be precise, this formula defines $K$ locally as a formal expansion near a chosen point of the diagonal. We require, however, that $K$ extends analytically globally and the determinantal identities hold globally.
	
	To any choice of a point $q\in\Sigma$ regular for all differentials and a local coordinate~$z$ at this point we can associate a formal power series $\tau(t_1,t_2,\dots)$ whose coefficients are defined by the power expansion of~$\omega_n$ in the chosen coordinate:
	\begin{equation}\label{eq:omega-def}
		\omega_n=\sum_{k_1,\dots,k_n=1}^\infty\frac{\partial^n\log(\tau)}{\partial t_{k_1}\dots\partial t_{k_n}}\Bigm|_{t=0}\prod_{i=1}^nz_i^{k_i-1}dz_i
		+\delta_{n,2}\frac{dz_1dz_2}{(z_1-z_2)^2},
	\end{equation}
	with the normalization $\tau(0)=1$.
	
	Then, \emph{the condition of KP integrability for a system of differentials $\omega_n$ is equivalent to the condition that $\tau$ is a KP tau function} (for some, and then, for any choice of~$q$ and~$z$), see~\cite{ABDKS3} and Section~\ref{sec:Hirota-proof} below. The latter condition can be expressed in the form of Hirota bilinear relations
	\begin{equation}\label{eq:Hirota-KP}
		\begin{gathered}
			\res_{z=0}e^{\sum_{k=1}^\infty(t_k-t'_k)z^{-k}}\tau(t-[z])\tau(t'+[z])\frac{dz}{z^2}=0,
			\\t\pm[z]\coloneq \bigl(t_1\pm\tfrac{z}{1},t_2\pm\tfrac{z^2}{2},t_3\pm\tfrac{z^3}{3},\dots\bigr),
		\end{gathered}
	\end{equation}
	that hold identically in $t=(t_1,t_2,\dots)$ and $t'=(t'_1,t'_2,\dots)$. Namely, the coefficient of any $(t,t')$-monomial in the expression under the residue is a Laurent $1$-differential in~$z$ with a pole of finite order at $z=0$, and we require that the coefficient of $\tfrac{dz}{z}$ of this differential vanishes.
	
	If the KP integrability holds, then the tau function~$\tau$ can also be expressed directly in terms of~$K$ by the following expansion in the basis of Schur functions $s_\lambda(t)$
	\begin{equation}
		\begin{aligned}
			\tfrac{K(z_1,z_2)}{\sqrt{dz_1}\sqrt{dz_2}}&=\sum K_{i,j}z_1^iz_2^j,\qquad |z_2|\ll|z_1|\ll 1,
			\\ \tau(t)&=\sum_{\lambda}\det\|K_{\lambda_i-i,j-1}\|_{i,j=1,\dots,\ell(\lambda)}\;s_\lambda(t).
		\end{aligned}
	\end{equation}
	The coefficients $K_{i,j}$ are also known as affine coordinates of the point of the Sato Grassmannian associated with a given tau function.
	
	\subsection{Example: trivial system of differentials}
	A KP integrable system of differentials $\{\omega_n\}$ is called trivial if $\omega_n=0$ for $n\ne2$. The differential $\omega_2$ cannot be equal to zero since it has a pole on the diagonal with biresidue~$1$. However, the KP integrability implies that it should be of a specific form.
	
	\begin{proposition}\label{prop:trivialKP}
		A system of symmetric differentials $\omega_n=\delta_{n,2}B$ is KP integrable if and only if the bidifferential $B=\omega_{2}$  is of the form
		\begin{equation}
			B(p_1,p_2)=\frac{dz(p_1)\,dz(p_2)}{(z(p_1)-z(p_2))^2}
		\end{equation}
		for some meromorphic function~$z$%
. The function~$z$ is defined uniquely up to a linear fractional transformation.
	\end{proposition}
	
This proposition is a reformulation of a result from~\cite{ABDKS3}. 
For a trivial KP integrable system of differentials and for a point~$q\in\Sigma$, one can find a local coordinate~$z$, %
such that $\tau_{q,z}\equiv1$, $\log(\tau_{q,z})\equiv0$. Note, however, that for a generic choice of local coordinate $z$ the tau function $\tau_{q,z}$ is not constant anymore but rather $\log(\tau_{q,z})$ is a nontrivial quadratic function.

Proposition is also valid in the case when $B$ is a series in $\hbar$ with meromorphic coefficients. In this case the function~$z$ and the local coordinate, trivializing the tau function, is also depending on~$\hbar$.
	
	If we require, in addition to triviality and KP integrability, that $\Sigma$ is compact and $\omega_{2}$ has no singularities apart from the diagonal, then $\Sigma$ is necessarily rational, and $\omega_{2}$ is the standard $\C P^1$ Bergman kernel,  $\omega_{2}=\frac{dz_1dz_2}{(z_1-z_2)^2}$, where $z$ is an affine coordinate on~$\Sigma=\C P^1$, see \cite{ABDKS3} for details.
	
	\subsection{Extended differentials}
	
	While discussing KP integrability for a system of differentials $\{\omega_n\}$, it is useful to consider an extended set of $n|n$ differentials (in fact, half-differentials) $\Omega^\bullet_n$ and their connected counterparts $\Omega_n$ defined by
	\begin{align} \label{eq:Omega-bullet-expint}
		\Omega^\bullet_n(p^+_{\set n},p^-_{\set n}) = 	&  \prod_{1\leq k<\ell\leq n} \frac{(z^+_k-z^+_\ell)(z^-_k- z^-_\ell)}{(z^+_k- z^-_\ell)(z^-_k-z^+_\ell)}
		\prod_{i=1}^n \frac{\sqrt{dz^+_i}\sqrt{dz^-_i}}{z^+_i- z^-_i}  \times
		\\ \notag
		&
		\exp\Bigg(
		\sum\limits_{\scriptsize\substack{m_1,\dots,m_n\geq 0 \\ \sum_{i=1}^n m_i=m\geq 1}} \prod\limits_{i=1}^n \frac{1}{m_i!} \Big(\int\limits_{ p^-_i}^{p^+_i} \Big)^{m_i} \bigl(\omega_m(\tilde p_{\set m})-\delta_{m,2}\tfrac{d\tilde z_1d\tilde z_2}{(\tilde z_1-\tilde z_2)^2}\bigr)
		\Bigg)  ,
		\\\label{eq:Omega-expint}	\Omega_n(p^+_{\set{n}},p^-_{\set n})  =& \sum_{\ell=1}^n \frac{(-1)^{\ell-1}}{\ell} \sum_{\scriptsize\substack{I_1\sqcup \cdots \sqcup I_\ell = \set n, \\ \forall j\, I_j\not=\emptyset}} \prod_{j=1}^\ell \Omega_{|I_j|}^{\bullet}(p^+_{I_j},p^-_{I_j}).
	\end{align}
	Here $p^+_{\set{n}}=(p^+_1,\dots,p^+_n)$ and $p^-_{\set{n}}=(p^-_1,\dots,p^-_n)$ are two sets of independent arguments, $z$ is an arbitrary local coordinate on~$\Sigma$, and $z^\pm_i=z(p^\pm_i)$, $\tilde z_i=z(\tilde p_i)$. These formulas define $\Omega^\bullet_n$, $\Omega_n$ in a neighborhood of some point on the diagonal. The point is such that these differentials do not depend on a choice of the local coordinate~$z$ and extend globally to~$\Sigma^{2n}$. Moreover, if the system of differentials $\{\omega_n\}$ is KP integrable, then the extended differentials satisfy the following determinantal identities (see \cite{ABDKS3})
	\begin{align}\label{eq:Omega-bullet-det}
		\Omega^\bullet_n(p^+_{\set n},p^-_{\set n})&=\det\|K(p^+_i,p^-_j)\|,\\
		\label{eq:Omega-det}\Omega_n(p^+_{\set n},p^-_{\set n})&=\det\nolimits^\circ\|K(p^+_i,p^-_j)\|,
	\end{align}
	where $\det\nolimits^\circ$ is a connected version of determinant. In particular,
	\begin{equation}
		\Omega_1^\bullet=\Omega_1=K.
	\end{equation}
	
	Notice that the pole on the diagonals $p_i^+=p_i^-$ cancel out for the connected $n|n$ differentials for $n\ne1$, and we have
	\begin{equation}
		\omega_n (p_{\set n}) = \restr{p^+_{\set n} }{p_{\set n}}\restr{p^-_{\set n}}{p_{\set n}}\Omega_n(p^+_{\set n}, p^-_{\set n}),\quad n\ge2.
	\end{equation}
	Therefore, the determinantal identities for~$\omega_n$ are specializations of those for~$\Omega_n$.
	
	It worth to mention that there is a direct combinatorial expression for $\Omega_n$ in terms of $\omega$-differentials as a sum over certain graphs~\cite{ABDKS3}.
	
	\subsection{VEV presentation}\label{Sec.VEV}
	We present here expressions for the differentials $\omega_n$, $\Omega^\bullet_n$, $\Omega_n$ in terms of the corresponding tau function as certain vacuum expectation values in the Fock space, see \cite{ABDKS3} for details. These expressions serve actually as a motivation for introducing these differentials.
	
	We consider the (bosonic) Fock space $\C[[t_1,t_2,\dots]]$ consisting of formal power series in~$t$-variables. Introduce the following operators acting in this space:
	\begin{align}
		J_k&=\begin{cases}\partial_{t_k},& k>0,\\0,&k=0,\\(-k)t_{-k},&k<0,\end{cases}
		\\X(z,\bar z)&=\frac{\sqrt{dz}\sqrt{d\bar z}}{z-\bar z}
		\;e^{\sum\limits_{k<0}\frac{z^k-\bar z^k}{k}J_k}e^{\sum\limits_{k>0}\frac{z^k-\bar z^k}{k}J_k},
		\\J(z)&=\restr{\bar z}{z}\left(X(z,\bar z)-\tfrac{\sqrt{dz}\sqrt{d \bar z}}{z-\bar z}\right)
		=\sum_{k=-\infty}^\infty z^{k-1}dz\,J_k.
	\end{align}
	Denote by $\vev{\cdot}$ the functional that takes the free term of the corresponding series, $\vev{\tau(t)}=\tau(0)$. Then, for the tau function $\tau=\tau_{q,z}$ associated with the local coordinate~$z$ at the point~$q\in\Sigma$ we have:
	\begin{align}
		\label{eq:Omega-bullet-VEV}
		\Omega^\bullet_n(p^+_{\set n},p^-_{\set n})&=\VEV{X(z^+_1,z^-_1)\dots X(z^+_n,z^-_n)\;\tau}
		\\\label{eq:Omega-bullet-tau}&=
		\prod_{1\leq k<\ell\leq n} \frac{(z^+_k-z^+_\ell)(z^-_k- z^-_\ell)}{(z^+_k- z^-_\ell)(z^-_k-z^+_\ell)}
		\prod_{i=1}^n \frac{\sqrt{dz^+_i}\sqrt{dz^-_i}}{z^+_i- z^-_i}
		\times\\\notag &\hskip3cm
		\tau(t)\bigm|_{t_k=\frac{1}{k}\sum\nolimits_{i=1}^n((z^+_i)^k-(z^-_i)^k)},
		\\\Omega_n(p^+_{\set n},p^-_{\set n})&=\VEVc{X(z^+_1,z^-_1)\dots X(z^+_n,z^-_n)\;\tau},
		\\\label{eq:omega-VEV}\omega_n(p_{\set n})&=\VEVc{J(z_1)\dots J(z_n)\;\tau},
	\end{align}
	where $z_i=z(p_i)$, $z^\pm_i=z(p^\pm_i)$, and where $\vev{\cdot}^\circ$ denotes the corresponding connected correlators related to the disconnected ones $\vev{\cdot}$ by inclusion-exclusion relations. Note that the singular correction for $n=2$ appearing in~\eqref{eq:omega-def} is included to~\eqref{eq:omega-VEV} automatically. Note also the following special case of~\eqref{eq:Omega-bullet-VEV}:
	\begin{equation}\label{eq:K-VEV}
		K(p,\bar p)=\VEVc{X(z,\bar z)\;\tau}=\frac{\sqrt{dz}\sqrt{d\bar z}}{z-\bar z}\;\tau(t)\bigm|_{t_k=\frac{z^k-\bar z^k}k},
		\quad z=z(p),~\bar z=z(\bar p).
	\end{equation}
	
	\begin{remark}
	The right hand sides in~\eqref{eq:Omega-bullet-VEV}--\eqref{eq:K-VEV} are defined as formal Laurent expansions in a suitable sector of $z$-variables but the left hand sides extend as global meromorphic objects on the whole $\Sigma^{2n}$ or~$\Sigma^{n}$, respectively. Therefore, these relations should be considered not as definitions of the left hand sides but rather as a motivation for their actual definitions~\eqref{eq:omega1K}--\eqref{eq:omeganK},~\eqref{eq:Omega-bullet-expint}--\eqref{eq:Omega-expint} or~\eqref{eq:Omega-bullet-det}--\eqref{eq:Omega-det}.
	\end{remark}
	
	\subsection{Expression in terms of theta function}\label{sec:Kr-constr}
	
	Assume that we are given the following data:
	\begin{itemize}
		\item a function (a power series) $\theta(w_1,w_2,\dots)$ with finite or infinite number of arguments such that $\theta(0)\ne0$;
		\item a collection of meromorphic differential $1$-forms $\eta=(\eta_1,\eta_2,\dots)$ on a (not necessarily compact) Riemann surface~$\Sigma$;
		\item a symmetric bidifferential~$B(p_1,p_2)$ on~$\Sigma^2$ with a pole on the diagonal with biresidue one.
	\end{itemize}
	Then, we can define
	\begin{equation}
		\begin{aligned}\label{eq:Kr-omega}
			\omega_n(p_{\set{n}})&=\prod_{i=1}^n(\eta(p_i)\partial_w)\log\theta(w)\bigm|_{w=0}+\delta_{n,2}B(p_1,p_2)
			\\&\coloneq \sum_{k_1,\dots,k_n}\frac{\partial^n\log \theta(w)}{\partial w_{k_1}\dots\partial w_{k_n}}\Bigm|_{w=0} \eta_{k_1}(p_1)\dots\eta_{k_n}(p_n)+\delta_{n,2}B(p_1,p_2).
		\end{aligned}
	\end{equation}
	Then, both the series $\tau=\tau_{q,z}$ associated with a choice of a regular point~$q\in\Sigma$ and a local coordinate~$z$ according to~\eqref{eq:omega-def} and the extended differentials $\Omega_n^\bullet$ for the differentials~$\{\omega_n\}$ can be expressed in terms of~$\theta$. Explicitly, define the coefficients $a_{i,j}$, $b_{i,j}$ by the following expansions at the point~$q$:
	\begin{align}
		\eta_i(p)&=\sum_{j=1}^\infty a_{i,j}z^{j-1}dz,\quad z=z(p),
		\\B(p_1,p_2)-\frac{dz_1dz_2}{(z_1-z_2)^2}&=\sum_{i,j=1}^\infty b_{i,j}z_1^{i-1}dz_1z_2^{i-1}dz_2,
		\quad z_i=z(p_i).
	\end{align}
	Then, we can introduce (independently of KP integrability)
	\begin{align}\label{eq:Kr-tau-1}
		\tau(t)&=e^{Q(t)}\frac{\theta(w)\bigm|_{w_k=T_k(t)}}{\theta(0)},
		\\\notag&T_k(t)=\sum_{\ell=1}^\infty a_{k,\ell}t_\ell,
		\\\notag&Q(t)=\frac12\sum_{i,j=1}^\infty b_{i,j}t_it_j,
		\\\Omega^\bullet_n(p^+_{\set n},p^-_{\set n})&=
		\prod_{1\leq k<\ell\leq n} \frac{E(p^+_k,p^+_\ell)E(p^-_k,  p^-_\ell)}{E(p^+_k, p^-_\ell)E(p^-_k,p^+_\ell)}
		\prod_{i=1}^n \frac{1
		}{E(p^+_i, p^-_i)}
		\frac{\theta\bigl(\textstyle\sum_{i=1}^n(\cA(p^+_i)-\cA(p^-_i))\bigr)}{\theta(0)}.
		\\\notag &\cA(p)=\int_q^p\eta=\Bigl(\int_q^p\eta_1,\int_q^p\eta_2,\dots\Bigr),
		\\\notag &E(p_1,p_2)=\frac{z_1-z_2}{\sqrt{dz_1}\sqrt{dz_2}}e^{-\frac12\int\limits_{p_2}^{p_1}\int\limits_{p_2}^{p_1}
			\bigl(B(\tilde p_1,\tilde p_2)-\frac{d\tilde z_1d\tilde z_2}{(\tilde z_1-\tilde z_2)^2}\bigr)},\quad z_i=z(p_i).
	\end{align}
	
	In the case when $\Sigma$ is a disc, $\theta=\tau$, $w_k=t_k$, $\eta_k=z^{k-1}dz$, and $B=\frac{dz_1dz_2}{(z_1-z_2)^2}$, the above expression for~$\Omega^\bullet_n$ is equivalent to~\eqref{eq:Omega-bullet-VEV}.
	Another important case is when $\Sigma$ is a compact genus~$g$ Riemann surface, $\theta=\theta(w_1,\dots,w_g)$ is the Riemann theta function, $\eta=(\eta_1,\dots,\eta_g)$ is a basis of holomorphic $1$-forms, and $B$ is the Bergman kernel, which is the unique bi-differential with the only pole being the double pole on the diagonal with bi-residue $1$ and 
%
vanishing~$\frak A$-periods (we refer to this Bergman kernel as the \emph{standard Bergman kernel} in what follows). In this case,~\eqref{eq:Kr-tau-1} is the KP tau function of Krichever construction, and the KP integrability for the differentials~\eqref{eq:Kr-omega} is a reformulation of Krichever theorem, see \cite{ABDKSnp}.
	
	\subsection{Proof of Hirota equations}\label{sec:Hirota-proof}
	Though equivalence of Hirota equations and determinantal identities can be considered as well known (in particular, under the name of the Wick theorem), we still present here computations deriving Hirota equations from determinantal identities. Variations of these arguments will be used in this paper elsewhere.
	
	Consider a collection of $2m+2n-2$ pairwise distinct points on $\Sigma$ regular for the kernel $K$:
	\begin{equation}
		Z=(p_1,\dots,p_m,\bar p_1,\dots,\bar p_{m-1},p'_1,\dots,p'_{n-1},\bar p'_1,\dots,\bar p'_{n}).
	\end{equation}
	
	\begin{lemma}\label{lem:Hirota-Omega}
		For a KP integrable system of differentials $\{\omega_n\}$ the following relations hold
		\begin{equation}\label{eq:Hirota-Omega}
			\sum_{q\in Z}\res\limits_{z=q}\Omega^\bullet_m(p_{\set m},\bar p_{\set {m-1}},z)\Omega^\bullet_n(p'_{\set {n-1}},z,\bar p'_{\set {n}})=0
		\end{equation}
		for all $m,n\ge1$ and for arbitrary collection of points forming~$Z$.
	\end{lemma}
	
	\begin{proof}
		Applying determinantal identities and the column or row expansions, respectively, of the corresponding determinants we represent the differential under the residue in the form
		\begin{equation}
			\sum_{i=1}^m\sum_{j=1}^n K(p_i,z) K(z,\bar p'_j)\Omega_{i,j}
		\end{equation}
		with $\Omega_{i,j}$ independent of~$z$. The $(i,j)$th summand has simple poles at $z=p_i$ and $z=\bar p'_j$, with opposite residues $-K(p_i,\bar p'_j)\Omega_{i,j}$ and $K(p_i,\bar p'_j)\Omega_{i,j}$, respectively. Therefore, the contribution of this summand to the total residue in~\eqref{eq:Hirota-Omega} vanishes.
	\end{proof}
	
	\begin{proof}[Proof of bilinear relations~\eqref{eq:Hirota-KP}]
		Let $\tau=\tau_{q,z}$ be the tau function associated with a point $q\in\Sigma$ and a local coordinate~$z$ at this point.
		Applying~\eqref{eq:Omega-bullet-tau} we find
		\begin{equation}
			\Omega^\bullet_m(p_{\set m},\bar p_{\set {m}})\bigm|_{p_m=q,\bar p_m=z}=
			C\,e^{\sum_{k=1}^\infty t_k z^{-k}}\tau(t-[z])\frac{\sqrt{dz}}{z}\bigm|_{t_k=\frac1k\sum_{i=1}^{m-1}(z_i^k-\bar z_i^k)},
		\end{equation}
		where $z_i=z(p_i)$, $\bar z_i=z(\bar p_i)$, and $C$ is the product of linear prefactors in~\eqref{eq:Omega-bullet-tau} not containing $\bar p_m=z$. A similar expression relating $\Omega^\bullet_n(p'_{\set n},\bar p'_{\set {n}})$ and~$\tau(t'+[z])$ holds as well. We see that after a suitable substitution of times the differential under the residue in~\eqref{eq:Hirota-KP} extends as a global differential coinciding with that one of~\eqref{eq:Hirota-Omega} up to a factor independent of~$z$ (in fact, we obtain a limit of~\eqref{eq:Hirota-Omega} as $\bar p_m\to q$, $p'_n\to q$, that is, $\bar z_m\to 0$, $z'_n\to 0$). Let us represent the residue in~\eqref{eq:Hirota-KP} as an integral over a contour surrounding the point $q\in\Sigma$. Then, choosing the coordinates $z_i,\bar z_i,z'_j,\bar z'_j$ of the points $p_i,\bar p_i,p'_j,\bar p'_j$ small enough we may assume that all these points belong to the disk bounded by the contour, and thus~\eqref{eq:Hirota-KP} reduces to~\eqref{eq:Hirota-Omega}.
		
		Thus, we derived not the original Hirota equations but their specializations under the substitution $t_k=\frac1k\sum_{i=1}^{m-1}(z_i^k-\bar z_i^k)$, $t'_k=\frac1k\sum_{i=1}^{n-1}(z'_i{}^k-\bar z'_i{}^k)$. This substitution is known as passing to the so called Miwa variables. It is representative enough, so that if some relation on tau functions holds in Miwa variables for any~$m$,~$n$, then it holds for the original times. Therefore, the original Hirota equations hold as well.
	\end{proof}
	
	\subsection{KP symmetries}
	
	In the local setting, KP symmetries act as linear transformations of the Fock space taking KP tau functions to KP tau functions. The group of KP symmetries is quite big: all KP tau functions form a single orbit of this action. It means that any KP tau function can be obtained from any other KP tau function, for example, from that one identically equal to~$1$, by the action of a suitable transformation from this group.
	
	The corresponding Lie algebra of infinitesimal KP symmetries is spanned by scalars and by the operators $E_{i,j}$, $i,j\in\Z$, defined by the following generating series:
	\begin{equation}
		\sum_{i,j=-\infty}^\infty z^{j}\bar z^{-i-1}E_{i,j}=\frac{
			e^{\sum\limits_{k<0}\frac{z^k-\bar z^k}{k}J_k}e^{\sum\limits_{k>0}\frac{z^k-\bar z^k}{k}J_k}-1}{z-\bar z}.
	\end{equation}
	
	In this section we describe the analogue of this Lie algebra of infinitesimal KP symmetries for the setting of global KP integrability. Namely, we describe infinitesimal transformations of a system of differentials $\{\omega_n\}$ that preserve infinitesimally all determinantal identities. Denote
	\begin{equation}\label{eq:cW-def1}
		\cW_n(q^+,q^-;p_{\set n})=\restr{p^\pm_{\set n}}{p_{\set n}}\restr{p^\pm_{n+1}}{q^\pm}\Omega_{n+1}(p^+_{\set{n+1}},p^-_{\set{n+1}}),
	\end{equation}
	where $\Omega_n$ is the extended $n|n$ halfdifferential associated with $\{\omega_n\}$ and defined by~\eqref{eq:Omega-bullet-expint}--\eqref{eq:Omega-expint}. There is also a direct combinatorial formula for these differentials:
	\begin{align} \label{eq:cT-omega}
		\cT_n(q^+,q^-;p_{\set n})&=\sum_{k=1}^\infty\frac1{k!}\biggl(\int\limits_{q^-}^{q^+}\biggr)^k
		\left(\omega_{k+n}(\tilde p_{\set{k}},p_{\set{n}})-
		\delta_{n,0}\delta_{k,2}\tfrac{d\tilde z_{1}d\tilde z_{2}}{(\tilde z_{1}-\tilde z_{2})^2}\right),
		\\ \label{eq:cW-def2}
		\cW_n(q^+,q^-;p_{\set n})&=\frac{\sqrt{dz^+}\sqrt{dz^-}}{z^+-z^-}\,e^{\cT_0(q^+,q^-)}
		\sum_{\substack{\set{n}=\sqcup_{\alpha} J_\alpha
				\\J_\alpha\ne\emptyset}}
		\prod_{\alpha}\cT_{|J_\alpha|}(q^+,q^-;p_{J_\alpha}),
	\end{align}
	where $z$ is an arbitrary local coordinate or just a meromorphic function, $\tilde z_i=z(\tilde p_i)$, and $z^\pm=z(q^\pm)$ (this definition is independent of a choice of~$z$).
	The integration in the first formula carries over $\tilde p_1,\dots,\tilde p_k$.
	
	\begin{remark} Note that \eqref{eq:cT-omega} (and, by proxy,~\eqref{eq:cW-def2}) contains infinite summation. It is well-defined in a formal expansion near a point on the diagonal, but often can be considered as a globally defined analytic object.
	\end{remark}
	
	\begin{proposition} \label{prop:infinitesimal-deformations}
		The following two-parameter family of infinitesimal transformations
		\begin{equation}
			\Delta\omega_n(z_{\set n})=\cW_n(q^+,q^-;p_{\set n})
		\end{equation}
		is an infinitesimal KP symmetry of a system of differentials $\{\omega_n\}$ for any specialization of parameters $(q^+,q^-)\in\Sigma^2$. More explicitly, assume that the given system of the differentials $\{\omega_n\}$ satisfy determinantal identities~\eqref{eq:omega1K}--\eqref{eq:omeganK}, and let
		\begin{equation}
			\begin{aligned}
				K_\epsilon(p_1,p_2)
				&=K(p_1,p_2)+\epsilon\,\Omega_2(p_1,q^+,p_2,q^-)+O(\epsilon^2)
				\\&=K(p_1,p_2)-\epsilon\,K(p_1,q^-)\,K(q^+,p_2)+O(\epsilon^2).
		\end{aligned}\end{equation}
		Then, we have
		\begin{align}
			\restr{p_2}{p_1}\left(K_\epsilon(p_1,p_2)-\frac{\sqrt{dz_1}\sqrt{dz_2}}{z_1-z_2}\right)&=\omega_1(p_1)+\epsilon\,\cW_1(q^+,q^-;p_1)+O(\epsilon^2),\quad z_i=z(p_i),\\
			\det\nolimits^\circ\|K_\epsilon(p_i,p_j)\|_{i,j=1,\dots,n}&=\omega_n(p_{\set n})+\epsilon\,\cW_n(q^+,q^-;p_{\set n})+O(\epsilon^2).
		\end{align}
		
	\end{proposition}

The proof is obtained by a straightforward expansion of the corresponding (connected) determinant.\qed

Specializations of parameters $(q^+,q^-)$ appearing in applications are often of the following sort. Let $\partial_x$ be a meromorphic vector field on~$\Sigma$. We denote by $p\mapsto \mathfrak{g}^t_{\partial_x}p$ the time~$t$ phase flow of this vector field. Then, we set $q^\pm=\mathfrak{g}^{\pm\frac{u\hbar}{2}}_{\partial_x}p$, take the coefficient of $u^k$ for some~$k$, multiply by some auxiliary meromorphic function, and integrate the result over some contour, see, e.g. \cite{alexandrov2024topologicalrecursionrationalspectral}. In the local description of Sect.~\ref{Sec.VEV}, the coefficients of the expansion of $\cW_n$ in $q^\pm$ generate the Lie algebra of symmetries of the KP hierarchy.
	
	\subsection{Multi-KP hierarchy}
		
	A system of KP integrable differentials provides solutions not only to the KP hierarchy but also to a more general $N$-KP hierarchy for arbitrary $N\ge1$. For $N=2$ it is known also as Toda lattice hierarchy \cite{UT}. In this section we review a construction for such solution. In the setting of Riemann theta function considered in Sect.~\ref{sec:Kr-constr}, it coincides with an algebraic-geometrical solution suggested in~\cite{krichever2023quasiperiodicsolutionsuniversalhierarchy}. The results of this section are not used in the proof of the main theorems of this paper but rather provide  their obvious applications.
	
	An $N$-KP tau function $\tau=\tau(\bs|\bt)$ depends on $N$ sets of continuous times $\bt_\alpha=(t_{\alpha,1},t_{\alpha,2},\dots)$, $\alpha=1,\dots,N$, and $N$ discrete times $\bs=(s_1,\dots,s_N)\in\Z^N$ satisfying
	\begin{equation}
		\sum_{\alpha=1}^Ns_\alpha=0.
	\end{equation}
	Equations of the $N$-KP hierarchy can be represented in the form of Hirota-type bilinear relations
	\begin{equation}\label{eq:Hirota-NKP}
		\sum_{\gamma=1}^N(-1)^{\sum\limits_{\alpha<\gamma}(s_\alpha{-}s'_\alpha)}
		\res\limits_{z=0} \frac{dz}{z^{s_\gamma{-}s'_\gamma}}
		e^{\sum\limits_{k=1}^\infty(t_{\gamma,k}-t'_{\gamma,k})z^{-k}}
		\tau(\bs{-}1_\gamma\mid\bt{-}[z]_\gamma)\;\tau(\bs'{+}1_\gamma\mid\bt'{+}[z]_\gamma)=0
	\end{equation}
	that hold for arbitrary sets of $t,t',s,s'$ variables satisfying $\sum s_\alpha=1$, $\sum s'_\alpha=-1$.
	Here~$1_k\in\Z^N$ is the $k$th basic unit vector, and $\bt\pm[z]_k$ denotes the corresponding shift in the $k$th set of times:
	\begin{equation}
		(\bt\pm[z]_\gamma)_{\alpha,i}=t_{\alpha,i}\pm\delta_{\alpha,\gamma}\frac{z^i}{i}.
	\end{equation}
	
	We provide a construction for a solution of $N$-KP hierarchy associated with a KP integrable system of differentials $\{\omega_n\}$, an arbitrary generic collection of $N$ pairwise distinct points $q_1,\dots,q_N$ on~$\Sigma$, and a choice of local coordinates $\zeta_1,\dots,\zeta_N$ at these points. It is convenient to identify the collection of discrete times with a degree zero divisor
	\begin{equation}
		D(\bs)=\sum_{\alpha=1}^N s_\alpha q_\alpha\in{\rm div}_0\Sigma.
	\end{equation}
	Our goal is to associate a tau function to such divisor. In order to unify notations, we denote this tau function by $\tau_D(\bt)=\tau(\bs|\bt)$, $D=\sum s_\alpha q_\alpha$.
	For the first step we construct a $D$-modified kernel~$K_D$. In the case when $D=\sum_{i=1}^m(q'_i-q''_i)$ for some collection of $2m$~pairwise distinct points $q'_1,\dots,q'_m, q''_1,\dots, q''_m$, it is given by an explicit formula
	\begin{equation}
		K_D(p,\bar p)=\frac{\Omega^\bullet_{m+1}(p,q'_{\set m},\bar p,q''_{\set m})}{\Omega^\bullet_{m}(q'_{\set m},q''_{\set m})},
		\quad D=\sum_{i=1}^m(q'_i-q''_i).
	\end{equation}
	Then, it proves out that $K_D$ admits a natural limit as some of the points $q'_i,q''_i$ glue together in a family. Moreover, this limit depends neither of the way we split the points of~$D\in{\rm div}_0\Sigma$ into pairs in order to represent it in the form $\sum_{i=1}^m(q'_i-q''_i)$, nor on the number $m$ of such pairs. Thus, by continuity, the definition of the kernel $K_D$ extends to arbitrary divisors of degree zero.
	
	The kernel $K_D$ defines, in turn, a system of $D$-modified differentials
	\begin{equation}
		\begin{aligned}
			\omega_{D,1}(p_1)&=\restr{p_2}{p_1}\left(K_D(p_1,p_2)-\frac{\sqrt{dz_1}\sqrt{z_2}}{z_1-z_2}\right),\quad z_i=z(p_i),\\
			\omega_{D,n}(p_{\set n})&=\det\nolimits^\circ\|K_D(p_i,p_j)\|_{i,j=1,\dots,n}.
		\end{aligned}
	\end{equation}
	
	In the case $D=\sum_{i=1}^m(q'_i-q''_i)$ with pairwise distinct $q'_i,q''_i$, the corresponding $n|n$ differentials are given by
	\begin{equation}
		\Omega^\bullet_{D,n}(p_{\set n},\bar p_{\set n})=\frac{\Omega^\bullet_{m+n}(p_{\set n},q'_{\set m},\bar p_{\set n},q''_{\set m})}
		{\Omega^\bullet_{m}(q'_{\set m},q''_{\set m})},
		\quad D=\sum_{i=1}^m(q'_i-q''_i).
	\end{equation}
	
	The kernel $K_D$ has poles at the points of the divisor~$D$. This implies that $\omega_{D,1}$ gets also poles at the corresponding points. Namely, $\omega_{D,1}$ has simple poles at $q_1,\dots,q_N$ with residues $s_1,\dots,s_N$, respectively. However, these poles cancel out for the higher differentials $\omega_{D,n}$, $n\ge2$, which means that all these differentials are, in fact, holomorphic for~$n\ge2$ except a standard pole of $\omega_{D,2}$ on the diagonal.
	
	The Taylor coefficients of the logarithm of $\tau_D$ are defined from the power expansion of the differentials $\omega_{D,n}$ at suitable points. Explicitly, for $n\ge1$ and an arbitrary collection of indices $\vec\alpha=(\alpha_1,\dots,\alpha_n)$, $\alpha_i\in\{1,\dots,N\}$, expanding $\omega_{D,n}$ at the point $(q_{\alpha_1},\dots,q_{\alpha_n})\in\Sigma^n$ in the local coordinates $\zeta_{\alpha_1},\dots,\zeta_{\alpha_n}$, respectively, we set
	\begin{equation}
		\frac{\partial^n\log\tau_D}{\partial t_{\alpha_1,k_1}\dots\partial t_{\alpha_1,k_1}}\Bigm|_{t=0}=
		[z_1^{k_1-1}\dots z_n^{k_n-1}]\frac{\tilde\omega_{D,n}(p_{\set n})}{\prod_{i=1}^n dz_i},
		\qquad z_i=\zeta_{\alpha_i}(p_i),\quad p_i\to q_{\alpha_i},
	\end{equation}
	where
	\begin{equation}
		\begin{aligned}
			\tilde \omega_{D,1}(p)&=\omega_{D,1}(p)-s_\alpha\frac{dz}{z},\quad \vec\alpha=(\alpha),\quad z=\zeta_\alpha(p),
			\\\tilde \omega_{D,2}(p_1,p_2)&=\omega_{D,2}(p_1,p_2)-\frac{dz_1dz_2}{(z_1-z_2)^2},
			\quad \vec\alpha=(\alpha,\alpha),\quad z_i=\zeta_\alpha(p_i),
		\end{aligned}
	\end{equation}
	and $\tilde\omega_{D,n}=\omega_{D,n}$ in all other cases (including the case when $n=2$ and $\vec\alpha=(\alpha_1,\alpha_2)$ with $\alpha_1\ne\alpha_2$). By construction, $\tilde\omega_{D,n}$ is indeed regular at the considered point and expands as a power series.
	
	This relation defines $\tau_D$ as a formal power series uniquely up to a constant factor $C_D=\tau_D(0)$. A correct account of this factor is quite sensitive for validity of Hirota equations. It is defined as follows.
	
	Let $U_\alpha$ be a small neighborhood of the point $q_\alpha$ where the local coordinate $\zeta_\alpha$ is defined. For any divisor $\widetilde D=\sum_{i=1}^m(q'_i-q''_i)$ with pairwise distinct $q'_i,q''_i$ and supported by $\cup_{\alpha=1}^N U_\alpha$ we define
	\begin{equation}
		C_{\widetilde D}=\prod_{1\le i<j\le m}\frac{\epsilon(q'_i,q''_j)\epsilon(q''_i,q'_j)}{\epsilon(q'_i,q'_j)\epsilon(q''_i,q''_j)}
		\prod_{i=1}^m\epsilon(q'_i,q''_i)
		~\Omega_m(q'_{\set m},q''_{\set m}),
	\end{equation}
	where for $p\in U_\alpha$, $\bar p\in U_\beta$, we define $\epsilon(p,\bar p)$ from relation
	\begin{equation}
		\epsilon(p,\bar p)\sqrt{d\zeta_\alpha(p)}\sqrt{d\zeta_\beta(\bar p)}=\begin{cases}
			\phantom{-}1,&\alpha<\beta,
			\\-1,&\alpha>\beta,
			\\\zeta_\alpha(p)-\zeta_\alpha(\bar p)&\alpha=\beta.
		\end{cases}
	\end{equation}
	The constant $C_{\widetilde D}$ does not depend on the way we split the points of $\widetilde D$ into pairs to represent it in the form $\sum_{i=1}^m(q'_i-q''_i)$ and extends by continuity to arbitrary degree~0 divisors supported by $\cup_{\alpha=1}^NU_\alpha$, in particular to the case of the divisor $D=\sum s_\alpha q_\alpha$. This completes the definition of the function $\tau_D(\bt)=\tau(\bs|\bt)$.
	
	\begin{theorem}
		The function $\tau_D(\bt)=\tau(\bs|\bt)$ is an $N$-KP tau function, that is, it satisfies the $N$-KP bilinear relations.
	\end{theorem}
	
	\begin{proof} Following the arguments in the proof of Hirota equations of Sect.~\ref{sec:Hirota-proof}, let us represent the sum of residues in~\eqref{eq:Hirota-NKP} as a single integral along a contour which is a boundary of the union of~$N$ small disks centered at $q_1,\dots,q_N$. Apply the Miwa substitution  $t_{\alpha,k}=\frac1k\sum_{i=1}^{m_\alpha}(z_{\alpha,i}^k-\bar z_{\alpha_i}^k)$, $t'_{\alpha,k}=\frac1k\sum_{i=1}^{n_\alpha}(z'_{\alpha,i}{}^k-\bar z'_{\alpha,i}{}^k)$, to the times in this relation, and also apply small deformations of the divisors $D=\sum s_\alpha q_\alpha$ and $D'=\sum s'_\alpha q_\alpha$ to represent the deformed divisors $\widetilde D$ and $\widetilde D'$ in the form $\widetilde D=\sum_{i=1}^{m_0}(\tilde q_i-\bar{\tilde q}_i)+p_m$, $\widetilde D'=\sum_{i=1}^{n_0}(\tilde q'_i-\bar{\tilde q}'_i)-\bar p'_n$, respectively. Then, the desired Eq.~\eqref{eq:Hirota-NKP} takes the form
		\begin{equation}\label{eq:Hirota-Omega-NKP}
			\oint\Omega^\bullet_m(p_{\set m},\bar p_{\set {m-1}},z)\Omega^\bullet_n(p'_{\set {n-1}},z,\bar p'_{\set {n}})=0,
		\end{equation}
		where $m=\sum_{\alpha=0}^N m_\alpha+1$, $n=\sum_{\alpha=0}^N n_\alpha+1$, and where the arguments of $\Omega^\bullet_m$ and $\Omega^\bullet_n$ include the points representing the Miwa variables and also the points of the divisors $\tilde D$ and~$\tilde D'$. These arguments are points of~$\Sigma$ distributed somehow among the disks centered at $q_1,\dots,q_N$. Moreover, we may assume that all these points are fixed except the one denoted by~$z$ which runs over the union of these disks. Assuming determinantal identities we obtain~\eqref{eq:Hirota-Omega-NKP} as a special case of Lemma~\ref{lem:Hirota-Omega}. Taking the limit as $\tilde D\to D$ and $\tilde D'\to D'$ we obtain the desired Hirota equations~\eqref{eq:Hirota-NKP} in Miwa coordinates. Finally, since the number of Miwa coordinates can be arbitrary, we conclude that the Hirota equations hold for the original times as well.
	\end{proof}
	
	\begin{remark}
	Note that if $s_k=0$ for all $k$, then the differentials $\omega_{D,n}$ coincide with the initial differentials $\omega_{n}$ and produce a tau-function $\tau({\bf 0}|\bt)$ that solves the KP hierarchy in each family of $t_{\alpha,i}$ for a given $\alpha$ and arbitrary values of $t_{\beta,i}$ for $\beta\neq \alpha$ that can be considered as parameters of the KP tau-function.
	\end{remark}
	
	\begin{example}
		Assume that we are in the setting of Sect. \ref{sec:Kr-constr}, that is, $\omega_n$ is given by~\eqref{eq:Kr-omega} for some choice of~$\theta$,~$\eta$, and~$B$. In this case the function $\tau_D$ can be expressed in terms of~$\theta$. These expressions presented below hold independently of KP integrability. However, if the differentials $\omega_n$ are KP integrable, then the function $\tau_D$ is a tau function of the $N$-KP hierarchy. For $D=\sum_{\alpha=1}^N s_\alpha q_\alpha$ we define
		\begin{align}
			E(p_1,p_2)&=\frac{z_1-z_2}{\sqrt{dz_1}\sqrt{dz_2}}e^{-\frac12\int\limits_{p_2}^{p_1}\int\limits_{p_2}^{p_1}
				\bigl(B(\tilde p_1,\tilde p_2)-\frac{dz_1dz_2}{(z_1-z_2)^2}\bigr)},\quad z_i=z(p_i),
			\\\cA_k(D)&=\sum_{\alpha=1}^Ns_\alpha \int_{q_0}^{q_\alpha}\eta_k,
			\\S_D(p)&=\sum_{\alpha=1}^Ns_\alpha \int_{q_0}^{q_\alpha}B(\cdot,p),
		\end{align}
		where~$z$ is an arbitrary meromorphic function, $q_0$ is a fixed point chosen in advance, and where we assume that the points $p_1,p_2,p$ and all integration contours belong to some fixed simply connected domain in~$\Sigma$ containing all points $q_1,\dots,q_N$. One can check that~$E$ and~$\cA_k(D),S_D$ do not depend on a choice of~$z$ and $q_0$, respectively.
		
		Define the constants $a_{k,(\alpha,i)}$, $b_{(\alpha,i),(\beta,j)}$, $c_{D,(\alpha,j)}$ by the expansions
		\begin{equation}
			\begin{aligned}
				\eta_k(p)&=\sum_{i=1}^\infty a_{k,(\alpha,i)}z^{i-1}dz,\quad z=\zeta_\alpha(p),~p\to q_\alpha,
				\\B(p,\bar p)-\delta_{\alpha,\beta}\frac{dz d\bar z}{(z-\bar z)^2}&=
				\sum_{i,j=1}^\infty b_{(\alpha,i),(\beta,j)}z^{i-1}d z\; \bar z^{j-1}d\bar z,\quad
				\begin{cases}
					z=\zeta_\alpha(p),&p\to q_\alpha,\\
					\bar z=\zeta_\beta(\bar p),&\bar p\to q_\beta,
				\end{cases}
				\\ S_D(p)-s_\alpha\frac{dz}{z}&=\sum_{j=1}^\infty c_{D,(\alpha,j)} z^{j-1}dz,\quad z=\zeta_\alpha(p),~p\to q_\alpha,
			\end{aligned}
		\end{equation}
		and set
		\begin{equation}
			\begin{aligned}
				T_{D,k}(t)&=\cA_k(D)+\sum_{\alpha=1}^N\sum_{i=1}^\infty a_{k,(\alpha,i)}t_{\alpha,i},
				\\
				Q_D(t)&=\frac12\sum_{\alpha,\beta=1}^N\sum_{i,j=1}^\infty b_{(\alpha,i),(\beta,j)}t_{\alpha,i}t_{\beta,j}
				+\sum_{\alpha=1}^N\sum_{j=1}^\infty c_{D,(\alpha,j)}t_{\alpha,j},
				\\
				\widetilde C_D&=\prod_{1\le\alpha<\beta\le N}\left(\frac{\epsilon(q_\alpha,q_\beta)}{E(q_\alpha,q_\beta)}\right)^{s_\alpha s_\beta},
			\end{aligned}
		\end{equation}

		Then, we have
		\begin{equation}\label{eq:Kr-tau}
			\tau_D(t)=\widetilde C_D\;e^{Q_D(t)}\theta(w)\bigm|_{w_k=T_{D,k}(t)}.
		\end{equation}
		
		This expression for~$\tau_D$ involves nontrivial finite shifts of times, and in order to give meaning to~\eqref{eq:Kr-tau} we need to impose some additional assumptions on~$\theta$. Namely, we assume that either the number of arguments $w_i$ of~$\theta$ is finite and~$\theta$ is analytic or that~$\theta$ is a series in an additional formal parameter~$\hbar$ and the coefficient of any power of $\hbar$ in~$\theta$ depends analytically on a finite number of $w$-variables.
		
		In the case when~$\Sigma$ is compact and $\theta$ is the Riemann theta function, Expression~\eqref{eq:Kr-tau} is equivalent to the one of~\cite{krichever2023quasiperiodicsolutionsuniversalhierarchy}.
	\end{example}
	
	\section{Convolution of two systems of differentials}
	
	\label{sec:Convolution}
	
	\subsection{Definition of convolved differentials}
	Let $\Sigma$ be a possibly noncompact Riemann surface (complex curve) and $\cP=\{q_1,\dots,q_N\}\subset\Sigma$ be a finite set. For two meromorphic $1$-forms $\psi$ and $\phi$ such that $\psi$ possibly has poles and $\phi$ is holomorphic at~$\cP$ we define
	\begin{align}
		\psi*\phi\coloneqq \sum_{q\in \cP} \res_{p=q} \psi(p) \int_q^p \phi.
	\end{align}

	Now consider two systems of symmetric differentials $\{\tilde \psi_n\}_{n\ge1}$ and $\{\tilde \phi_n\}_{n\ge1}$ on~$\Sigma$. We assume that both $\tilde \psi_n$ and $\tilde\phi_n$ expand as a series in~$\hbar$, $\tilde \psi_n = \sum_{d}\hbar^d\tilde\psi_n^{\langle d\rangle}$ and $\tilde \phi_n = \sum_{d}\hbar^d\tilde\phi_n^{\langle d\rangle}$, whose coefficients $\tilde\psi_n^{\langle d\rangle}$ and $\tilde\phi_n^{\langle d\rangle}$ are meromorphic and regular on diagonals. We use the tilde symbol to emphasize that regularity on diagonals is satisfied for all~$n$ including $n=2$. Later on we will modify $\tilde\psi_2$ and $\tilde\phi_2$ by introducing a pole on the diagonal for these differentials.
	
	We also assume that the only poles that $\tilde \psi_n$ may possess are those of the form $p_i=q_j$ where $p_1,\dots,p_n$ are the arguments of~$\tilde\psi_n$ and $q_j$ is a point of $\cP$. Assume also that $\tilde\phi_n$ are regular on~$\cP$.

	\begin{definition}\label{def:convolution-system}
		A system of differentials $\tilde\omega_n$ of convolution of the differentials  $\{\tilde \psi_n\}_{n\ge1}$ and $\{\tilde \phi_n\}_{n\ge1}$ is defined as follows.
		
		\begin{itemize}
			\item
			For $n\geq 1$ define a set $G_n$ of graphs $\Gamma$ with some additional structure and properties:
			\begin{itemize}
				\item The set of vertices $V(\Gamma)$ splits into three subsets $V(\Gamma)=V_\ell(\Gamma)\sqcup V_{\psi}(\Gamma)\sqcup V_{\phi}(\Gamma)$, which we call leaves, $\psi$-vertices and $\phi$-vertices, respectively.
				\item The set $V_\ell(\Gamma)$ consists of $n$ ordered leaves $\ell_1,\dots,\ell_n$, that is, there is a fixed bijection $\set n \to V_\ell(\Gamma)$ such that $i\mapsto \ell_i$.
				\item Each edge in the set of edges $E(\Gamma)$ either connects a leaf to a $\psi$- or a $\phi$-vertex (such edges are called \emph{external}), or connects a $\psi$-vertex to a $\phi$-vertex (such edges are called \emph{internal}). Accordingly, the set $E(\Gamma)$ splits as $E(\Gamma)=E_{\ell}(\Gamma)\sqcup E_{{\psi}-\phi}(\Gamma)$.
				\item There is exactly one edge attached to each leaf $\ell_i$.
				\item The graph $\Gamma$ is connected.
			\end{itemize}
			%
			%
			
			Here is an example of such graph; the $\psi$- and $\phi$-vertices are colored in white, resp., black; the internal edges are marked by a star:
			\[
			\vcenter
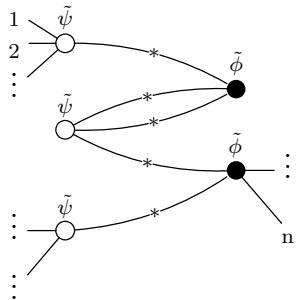
{
				\xymatrix@C=10pt@R=0pt{
					\txt{\tiny 1} \ar@{-}[dr] &   & \rule{30pt}{0pt} & &  &
					\\
					\txt{\tiny 2} \ar@{-}[r] &  *+[o][F-]{} \ar@{}[u]|(.8){\tilde\psi} \ar@{-}@/^5pt/[rrd]|{*}  & &   &  &
					\\
					\vdots \ar@{-}[ur] & & & *+[o][F*]{} \ar@{}[lu]_(.1){\tilde\phi} & &
					\\
					& *+[o][F-]{} \ar@{}[u]|(.6){\tilde\psi} \ar@{-}@/^5pt/[rru]|{*} \ar@{-}@/_5pt/[rru]|{*}  \ar@{-}@/_5pt/[drr]|{*}& &   & &
					\\
					&  & & *+[o][F*]{}  \ar@{}[lu]_(.1){\tilde\phi}  & \vdots \ar@{-}[l] &
					\\
					\vdots \ar@{-}[r]& *+[o][F-]{} \ar@{}[u]|(.4){\tilde\psi} \ar@{-}@/_5pt/[urr]|{*} & & &\txt{\tiny n} \ar@{-}[lu] &
					\\
					\vdots \ar@{-}[ur]
				}
			}
			\]
			
			\item To each graph $\Gamma\in G_n$ we associate its weight $\mathsf{w}(\Gamma)$ defined as follows.
			\begin{itemize}
				\item 
				A $\psi$-vertex (resp., $\phi$-vertex) $v$ of $\Gamma$ with $m$ edges attached to it is decorated by $\mathsf{w}(v)\coloneqq \tilde \psi_{m}$ (resp., $\mathsf{w}(v)\coloneqq \tilde\phi_{m}$), and the arguments of $\mathsf{w}(v)$ are put in bijection with the edges attached to $v$.
				\item 	The edge $e\in E_{\ell}(\Gamma)$ attached to $\ell_i$ is decorated by the operator $O(e)\coloneqq \restr {} {p_i}$.
				\item Each edge $e\in E_{{\psi}-\phi}(\Gamma)$ is decorated by the bilinear operator $O(e)\coloneqq *$ applied to the decorations of the $\psi$-vertex and $\phi$-vertex incident to $e$ in the variable corresponding to $e$.
				\item The weight $\mathsf{w} (\Gamma)$ is defined as
				\begin{align}
					\mathsf{w}(\Gamma) \coloneqq 
					\prod_{e\in E(\Gamma)} O(e)
					\bigg(\prod_{v\in V_{\psi}(\Gamma)} \mathsf{w}(v) \prod_{v\in V_{\phi}(\Gamma)} \mathsf{w}(v)\bigg),
				\end{align}
				where the action of the operators associated to the edges on the decorations of the vertices is prescribed by the graph $\Gamma$.
			\end{itemize}
			\item
			The convolved differentials $\tilde \omega_n$ are defined as
			\begin{align} \label{eq:GraphicalFormulaomega}
				\tilde \omega_n  =  \sum_{\Gamma\in G_{n}} \frac{\mathsf{w}(\Gamma)}{|\Aut(\Gamma)|},
			\end{align}	
			where $|\Aut(\Gamma)|$ denotes the order of the automorphisms group of a graph $\Gamma$ that preserves all markings of vertices.
		\end{itemize}
	\end{definition}
	
	\begin{remark} \label{rem:finiteness} Note that with this definition $\tilde \omega_n$ are given by an infinite summation over graphs. For the correctness of this definition we should impose some additional assumptions on $\tilde\psi_n$ and $\tilde \phi_n$ assuring that the number of graphs providing a nontrivial contribution to each $\tilde \omega^{\langle d \rangle}_n$ in $\tilde\omega_n=\sum_{d}\hbar^{d}\omega^{\langle d\rangle}_n$ is finite. For instance, it is sufficient to assume that the expansion of one of the two systems of differentials, say, $\{\tilde \psi_n\}$, in~$\hbar$ does not contain terms with $\hbar^{\leq 0}$, and the expansion of the other one, say, $\{\tilde \phi_n\}$, is regular in $\hbar$.
		
In particular, these expansions can be topological with extra vanishing of unstable terms. We remind the reader that, by definition, the expansion in $\hbar$ of a system of differentials, say $\{\tilde \psi_n\}$, is called \emph{topological}  if $\tilde \psi_n^{\langle d\rangle}\not=0$ only for $d = 2g-2+n$ where $g\in \Z_{\geq 0}$. By the \emph{unstable terms} of a topological expansion we mean $\tilde \psi^{\langle 0\rangle}_1$ and $\tilde \psi^{\langle 0\rangle}_2$.
	\end{remark}

	\begin{remark} \label{rem:topological-expansion}  Assume that the expansions of both $\tilde \psi_n$ and $\tilde \phi_n$ in $\hbar$ are topological and both $\tilde \psi^{\langle 0\rangle}_1$ and $\tilde \phi^{\langle 0\rangle}_1$ vanish, as well as at least one of two differentials $\tilde \psi_2^{\langle 0\rangle}$ and $\tilde \phi_2^{\langle 0\rangle}$. Then the convolution differentials $\tilde \omega_n$ are well defined and their expansion in $\hbar$ is topological as well.
	\end{remark}
	
	\begin{remark}
	If the convolution of all differentials of two collections vanishes, that is, $\psi_m*\phi_n=0$ for all $m,n\geq 0$, then the system of differentials $\tilde\omega_n$ of convolution of the differentials  $\{\tilde \psi_n\}$ and $\{\tilde \phi_n\}$ is nothing but their sum,
	\begin{equation}\
	\tilde\omega_n=\tilde \psi_n+\tilde \phi_n.
	\end{equation}
	\end{remark}
	
	
	\subsection{Expression in terms of potentials}\label{S_potentials}
	
	Assume that there is a sequence (finite or infinite) of meromorphic $1$-differentials $d\xi_1,d\xi_2,\dots$ such that for all $(d,n)$ the differential $\tilde\psi_n^{\langle d\rangle}$ admits an expansion of the form
	\begin{equation}
		\tilde\psi_n^{\langle d\rangle}(p_{\set n})=\sum_{k_1,\dots,k_n}
		c^{\langle d\rangle}_{k_1,\dots,k_n}d\xi_{k_1}(p_1)\dots d\xi_{k_n}(p_n)
	\end{equation}
	with some numerical coefficients $c^{\langle d\rangle}_{k_1,\dots,k_n}$ and with \emph{finitely many nonzero terms} for any fixed $(d,n)$. Then, all information on these coefficients, and hence, the differentials $\tilde\psi_n^{\langle d\rangle}$ can be encoded in a single generating series $F(t_1,t_2,\dots;\hbar)$ called \emph{potential} (cf.~\eqref{eq:Kr-omega})
	\begin{equation}
		F(t)=\sum_{n=1}^\infty\frac1{n!}\sum_{k_1,\dots,k_n}\sum_{d=0}^\infty \hbar^dc^{\langle d\rangle}_{k_1,\dots,k_n} t_{k_1}\dots t_{k_n}.
	\end{equation}
	
	In the setting of Definition~\ref{def:convolution-system}, assume that both systems of differentials $\{\tilde\psi_n\}$ and $\{\tilde \phi_n\}$ correspond to some potentials, that is, they are given by
	\begin{equation}
		\begin{aligned}
			\tilde\psi_n(p_{\set n})&=\restr{t}{0}\prod_{i=1}^n\Bigl({\textstyle\sum_k} d\xi_k(p_i)\frac{\partial}{\partial t_k}\Bigr) F^\psi(t),
			\\\tilde\phi_n(p_{\set n})&=\restr{s}{0}\prod_{i=1}^n\Bigl({\textstyle\sum_\ell} d\eta_\ell(p_i)\frac{\partial}{\partial s_\ell}\Bigr)F^\phi(s),
		\end{aligned}
	\end{equation}
	for two power series $F^\psi(t_1,t_2,\dots)$, $F^\phi(s_1,s_2,\dots)$ and two sequences of $1$-differentials $(d\xi_1,d\xi_2,\dots)$, $(d\eta_1,d\eta_2,\dots)$. By assumption, the differentials $d\xi_k$ may have poles at the points of~$\cP$ while $d\eta_\ell$ are holomorphic at these points.
	
	\begin{proposition} If both $\{\tilde\psi_n\}$ and $\{\tilde \phi_n\}$ admit expansions in terms of potentials, then, the corresponding convolution differentials $\tilde \omega_n$ also admit an expansion in terms of a potential,
		\begin{equation}
			\tilde\omega_n(p_{\set n})=\restr{t,s}{0}\prod_{i=1}^n\Bigl({\textstyle \sum_k} d\xi_k(p_i)\frac{\partial}{\partial t_k}
			+{\textstyle\sum_\ell} d\eta_\ell(p_i)\frac{\partial}{\partial s_\ell}\Bigr)F^\omega(t,s),
		\end{equation}
		and the potential $F^\omega$ is given explicitly by a Moyal-type product
		\begin{equation}
			\begin{aligned}
				e^{F^\omega(t,s)}&=e^{\sum_{k,\ell}b_{k,\ell}\frac{\partial^2}{\partial t_k\partial s_\ell}}e^{F^\psi(t)}e^{F^\phi(s)},
				\\ b_{k,\ell}&=d\xi_k*d\eta_\ell=\sum_{q\in\cP}\res\limits_{p=q}d\xi_k(p)\int_q^p d\eta_\ell.
			\end{aligned}
		\end{equation}
	\end{proposition}

	\subsection{Duality in the global setting}

	\begin{proposition} \label{prop:duality-convolution-global}
		Assume that $\Sigma$ is compact, the only possible poles of the differentials $\tilde\psi_n$ are those of the form $p_i=q$, $q\in\cP$, while the only possible poles of the differentials $\tilde \phi_n$ are those of the form $p_i=q$, $q\in\cP^\vee$, where $p_1,\dots,p_n$ are the arguments of $\tilde \psi_n$ or $\tilde\phi_n$, respectively, and $\cP$ and $\cP^\vee$ are two nonintersecting finite sets. Let residues and the $\mathfrak{A}$-periods of these differentials vanish for all $n$,
		\begin{equation}
		\res\limits_{p=q}\tilde{\psi}_n(p,p_{\set {n-1}})=0, \quad q\in\cP, \quad \quad \res\limits_{p=q}\tilde{\phi}_n(p,p_{\set {n-1}})=0, \quad q\in\cP^\vee,
		\end{equation}
		and
		\begin{equation}
			\oint_{p\in\mathfrak{A}_i}\tilde{\psi}_n(p,p_{\set {n-1}})=0, \quad \oint_{p\in\mathfrak{A}_i}\tilde{\phi}_n(p,p_{\set {n-1}})=0, \quad i=1,\dots,g(\Sigma),
		\end{equation}
		where $g(\Sigma)$ is the genus of the curve $\Sigma$.
		
		 Then, the convolution differentials of two systems of differentials $\{\tilde\psi_n\}$ and $\{\tilde\phi_n\}$ coincide with those of the systems  $\{\tilde\phi_n\}$ and $\{\tilde\psi_n\}$ (with the role of $\tilde\psi_n$ and $\tilde\phi_n$ swapped).
	\end{proposition}
	
	\begin{proof}
		Indeed, for two global differential $1$-forms $\psi$ and $\phi$ with possible poles at $\cP$ and $\cP^\vee$, respectively, we have
		\begin{equation}\label{eq:Star}
			\psi*\phi=\sum_{q\in\cP}\res\limits_{p=q}\psi(p)\int_q^p\phi
			=-\sum_{q\in\cP}\res\limits_{p=q}\phi(p)\int_q^p\psi
			=\sum_{q\in\cP^\vee}\res\limits_{p=q}\phi(p)\int_q^p\psi=\phi*\psi.
		\end{equation}
	\end{proof}
	
	\begin{remark}
		Note that the $*$-operation on the left hand side and the right hand side of Equation~\eqref{eq:Star} has different meaning: on the left hand side it is taken with respect to the set $\cP$, and on the right hand side it is taken with respect to the set $\cP^\vee$.
	\end{remark}
	
	\subsection{Conservation of KP integrability}
	
	Let $B_0$ be a symmetric bidifferential of the following special form
	\begin{equation}
		B_0(p_1,p_2)=\frac{dz(p_1)dz(p_2)}{(z(p_1)-z(p_2))^2},
	\end{equation}
	where $z$ is a fixed meromorphic function on~$\Sigma$. Set
	\begin{equation}
		\psi_n=\tilde\psi_n+\delta_{n,2}B_0,\quad
		\phi_n=\tilde\phi_n+\delta_{n,2}B_0,\quad
		\omega_n=\tilde\omega_n+\delta_{n,2}B_0,
	\end{equation}
	where the differentials $\{\tilde\omega_n\}$ are obtained by convolution of the systems of differentials $\{\tilde \psi_n\}$ and $\{\tilde\phi_n\}$.
	
	Note that the function $z$ defines a local embedding of $\Sigma$ into $\mathbb{C}P^1$. 	More precisely, we replace the curve $\Sigma$ by its open subset $U$ (with a smooth boundary $\partial U$) containing~$\cP$ and assume that $z$ used in the definition of $B_0$ defines an embedding of $U$ into $\C P^1$.
	Assume that $\{\tilde\psi_n\}$ extends via this embedding to the whole $\mathbb{C}P^1$ regularly outside the image of $U$ and $\{\tilde\phi_n\}$ is regular on $U$.

	\begin{theorem}\label{th:KP-conv}
		If the systems of differentials $\{\psi_n\}$ and $\{\phi_n\}$ are KP integrable and the assumptions on the regularity of $\{\tilde \psi_n\}$ outside $U$ (on $\mathbb{C}P^1$) and  of $\{\tilde \phi_n\}$ inside $U$ are satisfied, then the system of differentials $\{\omega_n\}$ is also KP integrable.
	\end{theorem}
	
	\subsection{Proof of Theorem~\ref{th:KP-conv}}
	
Let $\oint$ denote the sum of residues at the points of $\cP$.
	
	The proof of Theorem~\ref{th:KP-conv} is based on two fundamental properties of the convolution formula. In a nutshell, we have an interplay of the following two principles:
	\begin{itemize}
		\item Infinitesimal deformations of $\{\psi_n\}$ and $\{\phi_n\}$  that preserve KP integrability imply infinitesimal deformations of $\{\omega_n\}$ that preserve KP integrability;
		\item It is possible to fully trivialize $\{\psi_n\}$ or $\{\phi_n\}$ with the help of  deformations that preserve KP integrability. If $\{\psi_n\}$ (resp., $\{\phi_n\}$) is trivialized, that is, $\psi_n=\delta_{n,2}B_0$ (resp., $\phi_n=\delta_{n,2}B_0$), then $\{\omega_n\}=\{\phi_n\}$ (resp., $\{\omega_n\}=\{\psi_n\}$).
	\end{itemize}
	Once we formulate and prove the propositions realizing these two principles (see Proposition~\ref{prop:def-omegas} and Proposition~\ref{prop:deformations-KP} below, respectively), we are armed to complete the proof of Theorem~\ref{th:KP-conv}.
	
	\begin{remark} \label{rem:phi-psi-symmetry}
		Note that though $\psi$- and $\phi$-differentials play a bit different roles, they still enter the convolution formula for $\{\tilde\omega_n\}$ in a quite symmetric way, cf.~also Proposition~\ref{prop:duality-convolution-global}. So, in practice we have two possible choices for the exposition of the proof: we can choose which system of differentials we want to trivialize using the deformations preserving the KP integrability. The technical details are easier to follow in the case of trivialization of the $\phi$-differentials, so the argument below is presented for this case only.
	\end{remark}
	
	\subsubsection{Deformations preserving KP integrability}
	
	We consider a deformation problem for the differentials $\{\phi_n\}$ with respect to a small parameter $\epsilon$. In order to do this we need to enhance the notation:
	
	\begin{notation} Equations~\eqref{eq:cT-omega}--\eqref{eq:cW-def2} associate a system of differentials $\{\cW_n\}$ to a ``source'' system of differentials $\{\omega_n\}$. In an argument below we use the $\cW$-differentials for several different systems of the source differentials; in this case we enhance the notation by a superscript that refers to the source system. In particular, $\{\cW^\omega_n\}$ (resp., $\{\cW^\phi_n\}$), are the $\cW$-differentials associated to the source system $\{\omega_n\}$ (resp., $\{\phi_n\}$).
	\end{notation}

	Let $\phi_n(\epsilon; p_{\set n}) = \phi_n(p_{\set n}) + \epsilon \cW^\phi_n(p^+,p^-;p_{\set n})+O(\epsilon^2)$, where $p^+,p^-\in U$. In this case, Equation~\eqref{eq:omega_n-without-tilde} defines a system of differentials $\{\omega_n(\epsilon)\}$ that also depend on $\epsilon$.
	\begin{proposition} \label{prop:def-omegas} In this setting we have that $\omega_n(\epsilon; p_{\set n}) = \omega_n(p_{\set n}) + \epsilon \cW^\omega_n(p^+,p^-,p_{\set n})+O(\epsilon^2)$.
	\end{proposition}

	As a first step towards the proof of this proposition, we state an alternative formula for $\omega_n=\tilde \omega_n + \delta_{n,2}B_0$. Note that
	\begin{align} \label{eq:B0-phi}
		\tilde{\psi}_n(\cdot,p'_{\set{n-1}})*B_0(\cdot,p')=\oint_p \tilde{\psi}_n(\cdot,p'_{\set{n-1}})	\int_q^p B_0(\cdot ,p')  = \tilde{\psi}_n(p',p'_{\set{n-1}}),
	\end{align}
	where the second equality is a version of the projection property that is satisfied since $\{\tilde\psi_n\}$ extends to~$\C P^1$ regularly outside $U$.
	Now, consider a new set of graphs $G_n'$, which is exactly the same as $G_n$, but now all leaves are attached to the $\phi$-vertices. Moreover, in the definition of the weight $\mathsf{w}(G)$ for $\Gamma\in G_n'$ we use the weight $\mathsf{w}(v) = \tilde\phi_m + \delta_{m,2} B_0=\phi_m$ for the $\phi$-vertices of index $m$ except for the vertices of index two, where both edges are further attached to the $\psi$-vertices. The weight of the latter vertices remains unchanged and equal to $\tilde\phi_2= \phi_2-B_0$.
	
	\begin{lemma} \label{lem:Gamma-accent} We have:
		\begin{align} \label{eq:omega_n-without-tilde}
			\omega_n =  \sum_{\Gamma\in G'_{n}} \frac{\mathsf{w}(\Gamma)}{|\Aut(\Gamma)|}.
		\end{align}
	\end{lemma}
	\begin{proof}
		Indeed, in order to match the graphs in $G'_{n}$ and in $G_n$ we notice that
		\begin{align}
			\vcenter{\xymatrix@C=10pt@R=0pt{
					&  & \rule{20pt}{0pt} &  &
					\\
					& & & *+[o][F*]{}\ar@{}[u]|(.8){B_0} \ar@{-}@/_/[lld]|{*} \ar@{-}[r] & *+{\txt{\tiny $i$}}
					\\
					\vdots & *+[o][F-]{} \ar@{-}[l]+<5pt,7pt> \ar@{-}[l]+<5pt,-7pt> \ar@{-}[rr]+<-5pt,7pt>|{*} \ar@{-}[rr]+<-5pt,-7pt>|{*}
					\ar@{}[ru]^(.15){\tilde\psi} & & \vdots &
			}}
			=
			\vcenter{
				\xymatrix@C=10pt@R=0pt{
					{\txt{\tiny $i$}} \ar@{-}[rdd] &  & \rule{20pt}{0pt} &
					\\
					& & &
					\\
					\vdots & *+[o][F-]{} \ar@{-}[l]+<5pt,7pt> \ar@{-}[l]+<5pt,-7pt> \ar@{-}[rr]+<-5pt,7pt>|{*} \ar@{-}[rr]+<-5pt,-7pt>|{*}
					\ar@{}[ur]^(.15){\tilde\psi} & & \vdots
			}},
		\end{align}
		which is just Equation~\eqref{eq:B0-phi} applied inside weights of the graphs. The only exceptional contribution of graphs in $G_n'$ that does not have an interpretation in terms of graphs in $G_n$ is the one given by
		\begin{align}
			\vcenter{
				\xymatrix@C=10pt@R=0pt{
					& \txt{\tiny $1$} &
					\\
					*+[o][F*]{}\ar@{}[u]|(.8){B_0} \ar@{-}[ru] \ar@{-}[rd] &
					\\
					& \txt{\tiny $2$}   \ \
				}
			}
		\end{align}
		for $n=2$, which gives the extra summand $B_0(p_1,p_2)$. This implies that on the left hand side of Equation~\eqref{eq:omega_n-without-tilde} we have $\omega_n$ rather than $\tilde \omega_n$.
	\end{proof}
	
	With this lemma, we can proceed to the proof of the proposition.
	
	\begin{proof}[Proof of Proposition~\ref{prop:def-omegas}] Consider the sum over graphs that one obtains by plugging~\eqref{eq:omega_n-without-tilde} into Equations~\eqref{eq:cT-omega}--\eqref{eq:cW-def2}. It can be described as follows.
		
		\begin{itemize}
			\item
			For $n\geq 1$ define a set $G'_{\cW_n}$ of graphs $\Gamma$ with some additional structure and properties:
			\begin{itemize}
				\item The set of vertices $V(\Gamma)$ splits into three subsets $V(\Gamma)=V_\ell(\Gamma)\sqcup V_{\psi}(\Gamma)\sqcup V_{\phi}(\Gamma)$, which we call leaves, $\psi$-vertices and $\phi$-vertices, respectively.
				\item The set $V_\ell(\Gamma)$ consists of $n+1$ ordered leaves $\ell_1,\dots,\ell_n$, and $\ell_{\pm}$ that is, there is a fixed bijection $\set {n+1} \to V_\ell(\Gamma)$ such that $i\mapsto \ell_i$, $n+1\mapsto \ell_\pm$.
				\item Each edge in the set of edges $E(\Gamma)$ either connects a leaf to a $\phi$-vertex, or connects a $\psi$-vertex to a $\phi$-vertex. Accordingly, the set $E(\Gamma)$ splits as $E(\Gamma)=E_{\ell}(\Gamma)\sqcup E_{{\psi}-\phi}(\Gamma)$.
				\item There is exactly one edge attached to each leaf $\ell_i$, and an arbitrary number of edges attached to $\ell_\pm$.
				\item The graph $\Gamma$ is connected.
			\end{itemize}
			\item To each graph $\Gamma\in G'_{\cW_n}$ we associate its weight $\mathsf{w}_\cW(\Gamma)$ defined as follows.
			\begin{itemize}
				\item 
				A $\psi$-vertex (resp., $\phi$-vertex) $v$ of $\Gamma$ with $m$ edges attached to it is decorated by $\mathsf{w}(v)\coloneqq \tilde\psi_{m}$ (resp., $\mathsf{w}(v)\coloneqq \phi_{m}$), and the arguments of $\mathsf{w}(v)$ are put in bijection with the edges attached to $v$. The only exceptions are the $\phi$-vertices of index $2$: if both edges are connected to $\ell_\pm$, or its both edges are connected to $\psi$-vertices, then such $\phi$-vertex is decorated by $\mathsf{w}(v)\coloneqq\tilde \phi_{2}$.
				\item 	The edge $e\in E_{\ell}(\Gamma)$ attached to $\ell_i$ is decorated by the operator $O(e)\coloneqq \restr {} {p_i}$.
				\item 	Each edge $e\in E_{\ell}(\Gamma)$ attached to $\ell_\pm$ is decorated by the operator $O(e)\coloneqq \int_{p^-}^{p^+}$.
				\item Each edge $e\in E_{{\psi}-\phi}(\Gamma)$ is decorated by the bilinear operator $O(e)\coloneqq *$ applied to the decorations of the $\psi$-vertex and $\phi$-vertex incident to $e$ in the variable corresponding to $e$.
				\item The weight $\mathsf{w} _\cW(\Gamma)$ is defined as
				\begin{align}
					\mathsf{w}_\cW(\Gamma) \coloneqq 
					\prod_{e\in E(\Gamma)} O(e)
					\bigg(\prod_{v\in V_{\psi}(\Gamma)} \mathsf{w}(v) \prod_{v\in V_{\phi}(\Gamma)} \mathsf{w}(v)\bigg),
				\end{align}
				where the action of the operators associated to the edges on the decorations of the vertices is prescribed by the graph $\Gamma$.
			\end{itemize}
			\item
			The graphical interpretation of Equations~\eqref{eq:cT-omega}--\eqref{eq:cW-def2}, where we substituted $\omega$'s given by \eqref{eq:omega_n-without-tilde} reads:
			\begin{align} \label{eq:GraphicalFormula-cW-omega}
				\cW_n(p^+,p^-;p_{\set n})  =  \frac{\sqrt{dz(p^+)}\sqrt{dz(p^-)}}{z(p^+)-z(p^-)}\sum_{\Gamma\in G'_{\cW}} \frac{\mathsf{w}(\Gamma)}{|\Aut(\Gamma)|}.
			\end{align}	
		\end{itemize}
		
		Now we can apply the following resummation to~\eqref{eq:GraphicalFormula-cW-omega}. Namely, consider all $\phi$-vertices connected by an edge to $\ell_\pm$. Assume there are $m$ further edges attached to the union of these vertices. Fix the rest of the graph, including the attachments of these latter $m$ edges. Take the sum of the factors in the weight $\mathsf{w}_\cW$ over all possible subgraphs that can be formed by all $\phi$-vertices connected by an edge to $\ell_\pm$ with $m$ further edges attached to them. With an extra factor of $\frac{\sqrt{dz(p^+)}\sqrt{dz(p^-)}}{z(p^+)-z(p^-)}$ this sum gives a graphical interpretation of Equations~\eqref{eq:cT-omega}--\eqref{eq:cW-def2} for $\cW^\phi_m(p^+,p^-;p'_{\set m})$, where the last $m$ variables are put in correspondence with the $m$ further edges.
		
		Thus, the right hand side of Equation~\eqref{eq:GraphicalFormula-cW-omega} can be rewritten as a sum over $G'_n$, where to exactly one of the $\phi$-vertices we assign the weight $\cW^\phi_m(p^+,p^-)$ instead of $\phi_m$. The latter expression is equal to the coefficient of $\epsilon^1$ in~\eqref{eq:omega_n-without-tilde}, where the weight of each $\phi$-vertex of index $m$ is equal to $\phi_m(\epsilon)$. In other words, \eqref{eq:GraphicalFormula-cW-omega} is equal to the coefficient of $\epsilon^1$ in~\eqref{eq:omega_n-without-tilde} with the deformed $\phi$-differentials, which is precisely the statement of the proposition.
	\end{proof}
	
	\subsubsection{Deformations preserving KP integrability}
	Here we discuss the second principle mentioned above:
	\begin{proposition}\label{prop:deformations-KP} By specializations of the deformations given by
		\begin{align} \label{eq:deformations-cW}
			\phi_n(\epsilon; p_{\set n}) = \phi_n(p_{\set n}) + \epsilon \cW^\phi_n(p^+,p^-,p_{\set n})+O(\epsilon^2)	
		\end{align}	
		we span the full tangent space to the space of KP integrable differentials holomorphic on $U$ (outside the diagonal of $U^2$ for $n=2$, where $\phi_2$ must have a double pole with bi-residue $1$), and the latter space is connected.
	\end{proposition}
	\begin{proof} It is essentially a revised version of the statement of Proposition~\ref{prop:infinitesimal-deformations} from a different perspective. Indeed, every system of KP integrable differentials $ \{\phi_n\}$ such that $\phi_n-\delta_{n,2} B_0$ are holomorphic on $U^n$ (we call them ``regular'' KP integrable differentials) can be represented via determinantal formulas~\eqref{eq:omeganK} in terms of a 
		 kernel $K(z_1,z_2)$ such that $K(z_1,z_2)-\frac{\sqrt{dz_1}\sqrt{dz_2}}{z_1-z_2}$ is holomorphic on $U^2$ (we call such kernels ``regular'' kernels). Hence, the determinantal formulas induce an epimorphism from the tangent space to the space of ``regular'' kernels to the tangent space of the space of ``regular'' KP integrable differentials.
		
		Consider an arbitrary deformation of a ``regular'' $K(z_1,z_2)$ given by $K(\epsilon; z_1,z_2)=K(z_1,z_2)+\epsilon \Delta K(z_1,z_2)+O(\epsilon^2)$, where $\Delta K(z_1,z_2)$ is holomorphic in both arguments. Proposition~\ref{prop:infinitesimal-deformations} used in combination with
		\begin{align}
			\Delta K(z_1,z_2) = -\oint_{p_1'}\oint_{p_2'} \Delta K(z(p_1'),z(p_2')) K(z_1,z(p_1'))K(z(p_2'),z_2)
		\end{align}
		implies that the induced deformation of the differentials given by the determinantal formulas is given by
		\begin{align}
			\phi_n(\epsilon;p_{\set n}) = \phi_n (p_{\set n}) + \epsilon \oint_{p'_1}\oint_{p'_2} 	 \Delta K(z(p'_1),z(p'_2))\cW^\phi_n(p'_1,p'_2;p_{\set n}) + O(\epsilon^2),
		\end{align}
		which is indeed a specialization of the deformations~\eqref{eq:deformations-cW}.
		
		To conclude the proof, we only have to show that the space of ``regular'' KP integrable differentials is connected. To this end, notice that for an arbitrary ``regular'' kernel $K(z_1,z_2)$ the following formula
		\begin{align}
			K(t;z_1,z_2)\coloneqq (1-t)K(z_1,z_2)+t\tfrac{\sqrt{dz_1}\sqrt{dz_2}}{z_1-z_2}
		\end{align}
		for $0\leq t\leq 1$ connects $K(z_1,z_2)=K(0;z_1,z_2)$ through the space of ``regular'' kernels to $\tfrac{\sqrt{dz_1}\sqrt{dz_2}}{z_1-z_2}=K(1;z_1,z_2)$, and the determinantal formulas applied to $K(t;z_1,z_2)$ give a path in the space of ``regular'' KP integrable differentials connecting the given ones $\phi_n(p_{\set n})=\det\nolimits^\circ\|K(0;p_i,p_j)\|_{i,j=1,\dots,n}$ to $\delta_{n,2}B_0(p_1,p_2) = \det\nolimits^\circ\|K(1;p_i,p_j)\|_{i,j=1,\dots,n}$.
	\end{proof}
	
	\subsubsection{Final steps of the proof}
	
	Now we can proceed with the proof of the theorem.
	
	
	
	
	\begin{proof}[Proof of Theorem~\ref{th:KP-conv}]
		Proposition~\ref{prop:deformations-KP} implies that we can deform $\{\phi_n\}$ while preserving the KP integrability of the system, transforming it into $\{\overline{\phi}_n \coloneqq \delta_{n,2} B_0\}$. Simultaneously, the differentials $\{\omega_n\}$ undergo a corresponding deformation to a new set $\{\overline{\omega}_n\}$.
		
		Since, according to Proposition~\ref{prop:def-omegas}, the deformations of $\{\omega_n\}$ that accompany the deformations of $\{\phi_n\}$ also preserve KP integrability, it follows that the KP integrability of $\{\overline{\omega}_n\}$ implies the KP integrability of $\{\omega_n\}$.

On the other hand, in the case $\overline{\phi}_n - \delta_{n,2} B_0=0$, the contribution of all $\phi$-vertices in the graph summation formula of Definition~\ref{def:convolution-system} vanish and we have
$\{\overline{\omega}_n\} = \{\psi_n\}$. By assumption, $\{\psi_n\}$ is KP integrable, which implies that $\{\overline{\omega}_n\}$ is also KP integrable, and thus $\{\omega_n\}$ must be as well.
	\end{proof}

\section{Blobbed topological recursion}

\label{sec:Blobbed}

\subsection{Generalized topological recursion}\label{sec:GenTR}
Topological recursion of Chekhov--Eynard--Orantin \cite{CEO,EO-1st} associates a system of meromorphic differentials $\omega^{(g)}_n$, $g\geq 0$, $n\geq 1$, $2g-2+n\geq 0$, to an input data that consists of a Riemann surface $\Sigma$, a finite set of points $\cP\subset \Sigma$, two meromorphic functions $x$ and $y$ on $\Sigma$  such that the zeroes $\restr{}{q}dx = 0$ are simple and $\restr {} {q} {dy}\not= 0$ for each $q\in \cP$, and a bi-differential $B$ with the double pole on the diagonal with bi-residue $1$.
It has its origin in the computation of the cumulants of the matrix models, and by now it has multiple striking applications in algebraic geometry, enumerative combinatorics, and mathematical physics.

The original setup was generalized in~\cite{alexandrov2024degenerateirregulartopologicalrecursion} to the situation when~$dx$ and~$dy$ are meromorphic differentials with arbitrary behavior at the points of~$\cP$. We review here these definitions in the most general setting of local generalized topological recursion.

\begin{definition}\label{def:GendTR}
The initial data of the generalized topological recursion is a quintuple $(\Sigma,dx,dy,\cP,B)$ where
\begin{itemize}
\item $\Sigma$ is a (not necessarily compact) smooth Riemann surface (a complex curve);
\item $B$ is a symmetric meromorphic bidifferential on $\Sigma^2$ with a pole on the diagonal with biresidue~$1$ and holomorphic in the complement to the diagonal;
\item $\cP\subset\Sigma$ is a finite set; its elements are called \emph{key points};
\item $dx$ and $dy$ are meromorphic $1$-forms defined as germs (or even formal Laurent expansions) at the points of $\cP$.
\end{itemize}
\end{definition}

\begin{remark}
Note that we do not assume that $dx$ or $dy$ extend globally to $\Sigma$, but we require that $B$ is globally defined. On the other hand, we can assume that $\Sigma$ is a union of small disks centered at the key points. In this case $B$ is also defined locally near the key points only.	
\end{remark}

 Denote by $p\mapsto \mathfrak{g}^t_{\partial_x}p$ the time~$t$ phase flow of the vector field $\partial_x$ dual to the differential $dx$. Notice that for any meromorphic function~$f$ we have
\begin{equation}
f(\mathfrak{g}^t_{\partial_x}p)=e^{t\partial_x}f(p),\qquad x(\mathfrak{g}^t_{\partial_x}p)=x(p)+t.
\end{equation}

\begin{definition} We say that a system of differentials
\begin{equation}
\omega_n=\sum_{\substack{g\ge0\\(g,n)\ne(0,1)}}\hbar^{2g-2+n}\omega^{(g)}_n, \qquad n\geq 1,	
\end{equation}
  satisfies generalized topological recursion for the input data $(\Sigma,dx,dy,\cP,B)$ if
$\omega^{(0)}_2=B$ and the following two conditions hold:
\begin{itemize}
\item \emph{Generalized loop equations}: the differential
\begin{equation}\label{eq:loop-generalized}
\sum_{r\ge0}\bigl(-d\tfrac1{dy}\bigr)^r[u^r]e^{u(\cS(u\hbar\partial_x)-1)y}
\restr{q^\pm}{\mathfrak{g}^{\pm\frac{u\hbar}{2}}_{\partial_x}p}\cW_n(q^+,q^-;p_{\set n})
\end{equation}
is holomorphic in $p$ at any key point $q\in\cP$ for all $p_1,\dots,p_n$, where $x=x(p)$, and $y=y(p)$, $u$ is a new formal variable, and the extended differentials $\cW_n$ are defined by~\eqref{eq:cW-def1}--\eqref{eq:cW-def2}.
\item \emph{Projection property}: we have
\begin{equation}\label{eq:proj}
\omega^{(g)}_{n}(p,p_{\set {n-1}})=\sum_{q\in\cP}\res_{\tilde p=q} \omega^{(g)}_{n}(\tilde p,p_{\set {n-1}})\int_q^{\tilde p}B(\cdot,p),
\qquad 2g-2+n>0.
\end{equation}
\end{itemize}
\end{definition}

\begin{remark} These two conditions indeed allow to reconstruct the differentials $\omega^{(g)}_n$ recursively. Namely, the loop equations determine uniquely the principal parts of the poles of~$\omega^{(g)}_n$ at the key points with respect to the first argument inductively from the known differentials $\omega^{(g')}_{n'}$ with $2g'-2+n'<2g-2+n$, while the projection property recovers $\omega^{(g)}_n$ from its polar parts.
\end{remark}

\begin{remark}
 The resulting differentials $\omega^{(g)}_n$ are symmetric, meromorphic, and the only poles they may have are the key points in each argument.

In the case when the spectral curve data is nondegenerate, that is, $dx$ has simple zeroes and $dy$ is holomorphic and nonvanishing at the key points, the differentials of the generalized topological recursion coincide with those of usual CEO recursion.	
\end{remark}

\begin{remark}
It is sometimes convenient to define $\omega^{(0)}_1=y\,dx$, for some $y$ defined as a local primitive of $dy$. However, this differential is not always globally defined and univalued. Therefore, we prefer not to include the $(0,1)$-term to $\omega_1$, by convention. Instead, the contribution of the differentials $dx$ and $dy$ to the generalized loop equations (and, as a consequence, to the relations of the recursion) is included explicitly in~\eqref{eq:loop-generalized}.

On the other hand, the unstable term $\omega^{(0)}_2=B$ is included to $\omega_2$. Notice that the projection property is not satisfied for $(g,n)=(0,2)$: the right hand side of~\eqref{eq:proj} gives zero in this case.
\end{remark}

\begin{remark}
Let us write at a key point~$q\in\cP$ in some local coordinate $z$ centered at $q$
\begin{equation}
dx=a\,z^{r-1}dz\,(1+O(z)),\qquad
dy=b\,z^{s-1}dz\,(1+O(z)),\qquad a\,b\ne0.
\end{equation}
The point $q$ is called \emph{special} if $r+s>0$ and $(r,s)\ne(1,1)$, and \emph{non-special} otherwise. Then the non-special key points do not contribute to the topological recursion and they can be excluded from~$\cP$ without changing the differentials of topological recursion. Thus, one can assume without loss of generality that all key points are special.
\end{remark}

\subsection{Changing the set of key points}

Assume that the set $\cP$ of key points in the initial data of topological recursion splits as $\cP=\cP'\sqcup\cP''$.  Denote by $\psi_n$, $\phi_n$ and $\omega_n$ the differentials of (generalized) topological recursion with the spectral curve data $(\Sigma,dx,dy,\cP',B)$, $(\Sigma,dx,dy,\cP'',B)$, and $(\Sigma,dx,dy,\cP,B)$, respectively, that is, we use the same differentials $dx,dy,B$ but different sets of key points, see \cite{alexandrov2024degenerateirregulartopologicalrecursion} for details.

\begin{proposition}\label{prop:cP-cut} Assume that the curve $\Sigma$ is compact and $B$ is the standard Bergman kernel. 
The differentials $\tilde\omega_n=\omega_n-\delta_{n,2}B$ coincide with the convolution differentials of the systems of differentials $\tilde\psi_n=\psi_n-\delta_{n,2}B$ and $\tilde\phi_n=\phi_n-\delta_{n,2}B$.
\end{proposition}

This proposition means that for computing differentials of topological recursion it is sufficient to apply recursion itself in those cases with a single key point only. Then the TR differentials for the case of several key points are expressed in terms of those with one key point by explicit combinatorial expressions.

For the reformulation in terms of potential given in Sec. \ref{S_potentials}, it means that the potential for the generalized TR can be given as a Moyal-type product of the potentials associated with the key points. This is a reformulation of the Givental decomposition, there the diagonal part of the Givental operators is included into the definitions of the elementary potentials, and only the non-diagonal part describes interaction between different key points. Moreover, here we naturally can consider types of decompositions associated with different collections of the key points, which is close to the observations of \cite{AMM1,AMM2} in the context of matrix models and Virasoro constraints.

\begin{proof} Note that the projection property for $\{\omega_n\}$ follows manifestly from the projection properties for $\{\psi_n\}$ and $\{\phi_n\}$, and the structure of the convolution formula.
	
In order to prove the generalized loop equations at $q\in \cP''$ we use the same trick as in the proof of Proposition~\ref{prop:def-omegas}. Namely, we notice that $\cW^\omega_n(p^+,p^-;p_{\set n})$ can be represented as a linear combination of operators applied to
$\{\cW^\phi_m(p^+,p^-;p'_{\set m})\}$ that act only in the variables associated to $p'_1,p'_2,\dots$. Hence, the expression~\eqref{eq:loop-generalized}  in the generalized loop equations for $\{\omega_n\}$ is represented as a linear combination of operators applied to
\begin{equation}\label{eq:loop-generalized-phi}
	\sum_{r\ge0}\bigl(-d\tfrac1{dy}\bigr)^r[u^r]e^{u(\cS(u\hbar\partial_x)-1)y}
	\restr{q^\pm}{\mathfrak{g}^{\pm\frac{u\hbar}{2}}_{\partial_x}p}\cW^\phi_n(q^+,q^-;p'_{\set m})
\end{equation}
that act only on $p'_1,p'_2,\dots,p'_m$. Note that these expressions are finite for each fixed $n$ and degree of $\hbar$ (cf.~Remark~\ref{rem:topological-expansion}). Thus, the generalized loop equations for $\{\phi_n\}$ imply the generalized loop equations for $\{\omega_n\}$ at the points $q\in\cP''$.

In order to deal with the points $q\in\cP'$ we note that the combinatorics of graphs in the argument of Proposition~\ref{prop:def-omegas} works also with the interchanged roles of the differentials $\{\psi_n\}$ and $\{\phi_n\}$. To this end, we recall Proposition~\ref{prop:duality-convolution-global} --- we have a global setup, and  the residues of $\psi_n$ and $\phi_n$ indeed vanish, since these differential are obtained by topological recursion, so this proposition is applicable. Thus we can change the convolution operation on the edges to use the residues at the points of $\cP''$.  Hence, the argument above with the roles of $\{\psi_n\}$ and $\{\phi_n\}$ interchanged proves the generalized loop equations for $\{\omega_n\}$ at the points $q\in\cP'$.
\end{proof}

\begin{remark} It is important to stress that Proposition~\ref{prop:cP-cut} does not necessarily hold for a non-compact spectral curve. An important requirement is that we can change the convolution operation from the residues at the points of $\cP'$ to the points of $\cP''$, and this requires that the projection formula generates no singularities of the analytic extension of the involved differentials. In short, we can apply Proposition~\ref{prop:cP-cut} with locally defined $dx$ and $dy$ once we know that $B$ extends to the standard Bergman kernel on a global curve.
\end{remark}

\subsection{Definition of the blobbed topological recursion}\label{sec:TR-blobbed}

We are in the setting of Section~\ref{sec:GenTR}, that is, $\Sigma$ is a possibly noncompact Riemann surface and $\cP\subset\Sigma$ is a finite set.

Generalizing the situation of Proposition~\ref{prop:cP-cut}, consider an arbitrary system of symmetric differentials $\{\phi_n\}_{n\geq 1}$ called \emph{blobs}, where $\phi_n$ expands as a formal power series in $\hbar$, $\phi_n = \sum_{d=0}^\infty \hbar^d \phi_n^{\langle d\rangle}$. We assume that all $\phi_n^{\langle d \rangle}$ are holomorphic on $\Sigma^n$ except for $\phi_2^{\langle 0 \rangle}$, whose only singularity is the double pole on the diagonal with bi-residue $1$.

\begin{definition}\label{def:blobbedTR}
	
We associate to the input data $(\Sigma,dx,dy,\cP,\{\phi_n\}_{n\geq 1})$ a system of meromorphic differentials $\{\omega^{\langle d\rangle}_n\}$, $d\geq 0$, $n\geq 1$, which we call the \emph{blobbed TR differentials}, defined as follows.

\begin{itemize}
	\item
First, we choose an arbitrary Bergman kernel $B$ on $\Sigma^2$. Let $\{\psi^{B,(g)}_n\}$, $g\geq 0$, $n\geq 1$, $2g-2+n\geq 0$, be the system of differentials of generalized topological recursion associated to the input data $(\Sigma,dx,dy,\cP,B)$. Denote
\begin{align}\label{eq:89}
	\psi^{B}_n\coloneqq \sum_{\substack{g\ge0\\(g,n)\ne(0,1)}} \hbar^{2g-2+n}\psi^{B,(g)}_n.
\end{align}
\item Second, set
\begin{equation}\label{eq:90}
\tilde\psi^B_n=\psi^{B}_n-\delta_{n,2}B,\qquad
\tilde\phi_n=\phi_n-\delta_{n,2}B.
\end{equation}
\item
Denote by $\{\tilde\omega_n\}$ the convolution differentials of the systems of differentials $\{\tilde\psi^B_n\}$ and $\{\tilde\phi_n\}$.
The blobbed TR differentials $\omega_n=\sum_{d=0}^\infty\hbar^d \omega_n^{\langle d\rangle}$ are defined as
\begin{align} \label{eq:GraphicalFormulaomega}
	\omega_n  = \tilde\omega_n + \delta_{n,2} B.
\end{align}	
\end{itemize}
\end{definition}

\begin{remark} Note that with this definition the expansions of $\tilde\psi_n^B$ in~$\hbar$ contain no $\hbar^0$ terms, and hence, $\omega_n$ are well defined: there are finitely many graphs of Definition~\ref{def:convolution-system} contributing to each particular~$\omega_n^{\langle d\rangle}$.
\end{remark}

\begin{remark}
Note also that $\omega^{\langle 0 \rangle}_1 = \phi^{\langle 0 \rangle}_1$ and $\omega^{\langle 0 \rangle}_2 = \phi^{\langle 0 \rangle}_2$.
\end{remark}

\begin{remark} Assume that $\phi_n =\delta_{n,2} B$. Then $\omega_n = \psi^{B}_n$ are differentials of the usual (generalized) topological recursion.
\end{remark}

\begin{remark} Assume that the expansions of $\phi_n$ in $\hbar$ are topological, that is, $\phi_n^{\langle d\rangle}\not=0$ only if $d = 2g-2+n$ for some $g\geq 0$. Then the expansion of $\omega_n$ in $\hbar$ is topological as well. Exactly this situation was treated in details in~\cite{BS-blobbed} (in the setting of the original CEO topological recursion).
\end{remark}

\begin{example}
For arbitrary spectral curve data $(\Sigma,dx,dy,\cP,B)$ of topological recursion and a subset $\cP'\subset\cP$, the differentials $\omega_n$ of topological recursion with this spectral curve data can be regarded as differentials of blobbed topological recursion with the initial data~$(\Sigma,dx,dy,\cP',\{\phi_n\})$ with a smaller set $\cP'$ of key points and the differentials $\phi_n$ of topological recursion with the spectral curve data $(\Sigma,dx,dy,\cP\setminus\cP',B)$ regarded as blobs. This assertion is a reformulation of Proposition~\ref{prop:cP-cut}.
\end{example}

\subsection{Independence of a choice of~$B$}

\begin{lemma}\label{lem:Independence-of-B}
The blobbed differentials $\omega^{\langle d\rangle}_n$ are symmetric meromorphic differentials that do not depend on the particular choice of $B$.
\end{lemma}

\begin{example}
A special case of this lemma was observed in \cite[Section 5]{BS-blobbed}. Namely, assume that we have two possible choices for the Bergman kernel, $B_0$ and $B$. Then, the differentials of topological recursion computed with respect to the kernel $B$ (with the trivial blobs) coincide with the differentials of blobbed topological recursion associated with the differential $B_0$ as the Bergman kernel and the differentials $\tilde\phi_n=\delta_{n,2}(B-B_0)$ as the blobs.
\end{example}

The fact that $\{\omega_n\}$ are symmetric and meromorphic follows directly from the definition of the convolution of the systems of differentials. The proof that they do not depend on the choice of $B$ is given in the next Section. 

\subsubsection{Proof of Lemma~\ref{lem:Independence-of-B}} \label{sec:proof-Independence-B}

First, we state a deformation formula for the differentials of the topological recursion.

\begin{proposition}
	Let $\{\psi^{(g)}_n\}$ be the differentials of generalized topological recursion corresponding to the initial data $(\Sigma,dx,dy,\cP,B)$. Consider a deformation of the initial data with respect to a small parameter $\epsilon$ such that $\Sigma$, $dx$, $dy$ and $\cP$ are fixed and $B(\epsilon )=B+\epsilon \Delta B +O(\epsilon^2)$,  where~$\Delta B$ is a symmetric bidifferential holomorphic at $\cP\times\cP$. Let $\{\psi^{(g)}_n(\epsilon)\}$ be the differentials of generalized topological recursion corresponding to the initial data $(\Sigma,dx,dy,\cP,B(\epsilon))$. Then $\{\psi^{(g)}_n(\epsilon)=\psi^{(g)}_n+ \epsilon \Delta\psi^{(g)}_n+O(\epsilon^2) \}$, where
	\begin{align} \label{eq:Deformation-B}
		\Delta\psi^{(g)}_n(p_{\set n}) & =
		\sum_{i=1}^n\sum_{q\in\cP}\res\limits_{p=q}\psi^{(g)}_n(p,p_{\set{n}\setminus\{i\}})\int\limits_q^p\Delta B(\cdot,p_i)
		\\ \notag & \quad +
		\frac12\sum_{q,q'\in\cP}\res\limits_{p=q}\res\limits_{p'=q'}
		\biggl(\psi^{(g-1)}_{n+2}(p,p',p_{\set n})\biggr.
		\\ \notag & \qquad \qquad \qquad +\biggl.\sum_{\substack{g_1+g_2=g \\ I_1\sqcup I_2=\set{n}\\2g_i-2+|I_i|+1>0}}
		\psi^{(g_1)}_{|I_1|+1}(p,p_{I_1})\psi^{(g_2)}_{|I_2|+1}(p',p_{I_2})\biggr)
		\int\limits_q^p\int\limits_{q'}^{p'}\Delta B.
	\end{align}  	
\end{proposition}

\begin{proof}
	This formula is proved in the context of the original CEO topological recursion in~\cite{EO-1st}, see also an independent proof in~\cite{Kazarian} and an exposition in~\cite[Appendix]{alexandrov2024topologicalrecursionrationalspectral}. It the context of generalized topological recursion we can reduce this statement to the setting of the original CEO topological recursion by the deformation arguments of~\cite[Proof of Theorem 6.4]{alexandrov2024degenerateirregulartopologicalrecursion}. In a nutshell, the idea is to apply the following system of steps:
	\begin{itemize}
		\item Rewrite Equation~\eqref{eq:Deformation-B} replacing the residues at the points $q,q'\in \cP$ by the integrals over the boundary of some system of disks $D_q$, $q\in \cP$:
		\begin{align} \label{eq:Deformation-B-1}
			\Delta\psi^{(g)}_n(p_{\set n}) & =
			\sum_{i=1}^n\sum_{q\in\cP}\oint_{p\in \partial D_q}\psi^{(g)}_n(p,p_{\set{n}\setminus\{i\}})\int\limits_q^p\Delta B(\cdot,p_i)
			\\ \notag & \quad +
			\frac12\sum_{q,q'\in\cP}\oint_{p\in \partial D_q}\oint_{p\in \partial D_{q'}}
			\biggl(\psi^{(g-1)}_{n+2}(p,p',p_{\set n})\biggr.
			\\ \notag & \qquad \qquad \qquad +\biggl.\sum_{\substack{g_1+g_2=g \\ I_1\sqcup I_2=\set{n}\\2g_i-2+|I_i|+1>0}}
			\psi^{(g_1)}_{|I_1|+1}(p,p_{I_1})\psi^{(g_2)}_{|I_2|+1}(p',p_{I_2})\biggr)
			\int\limits_q^p\int\limits_{q'}^{p'}\Delta B.
		\end{align}
		\item Consider an additional deformation of the initial data with respect to a small parameter $\delta$ as in~\cite[Proof of Theorem 6.4]{alexandrov2024degenerateirregulartopologicalrecursion}. That is, $\Sigma$, $dy$, $B(\epsilon)$ are fixed. The differential $dx(\delta)$ is a local deformation of $dx$ inside $\sqcup_{q\in \cP}D_q$, where $dx(\delta)|_{D_q} =  dx+\delta\sum_{q\in\cP} \frac{dz_q}{(z_q-q)^{A_q}}$. Here $z_q$ is a local coordinate on $D_q$ and $A_q\in Z_{>0}$ is strictly bigger than the order of zeroes of $y$ at $q$. Finally, $\cP(\delta)$ is the set of all simple critical points of $dx(\delta)$ in $\sqcup_{q\in \cP}D_q$, which is the same as the set of all \nonregular{} points in this case.
		\item We get a two-parameter deformation family that analytically depends on $\epsilon$ and $\delta$. Thus we can interchange $\lim_{\delta\to 0}$ and $\restr{\epsilon}{0}\partial_\epsilon$.
		\item With this setup for $\delta\not=0$ the equation
		\begin{align} \label{eq:Deformation-B-2}
			\Delta\psi^{(g)}_n(\delta|p_{\set n}) & =
			\sum_{i=1}^n\sum_{q\in\cP}\oint_{p\in \partial D_q}\psi^{(g)}_n(\delta|p,p_{\set{n}\setminus\{i\}})\int\limits_q^p\Delta B(\cdot,p_i)
			\\ \notag & \quad +
			\frac12\sum_{q,q'\in\cP}\oint_{p\in \partial D_q}\oint_{p\in \partial D_{q'}}
			\biggl(\psi^{(g-1)}_{n+2}(\delta|p,p',p_{\set n})\biggr.
			\\ \notag & \qquad \qquad \qquad +\biggl.\sum_{\substack{g_1+g_2=g \\ I_1\sqcup I_2=\set{n}\\2g_i-2+|I_i|+1>0}}
			\psi^{(g_1)}_{|I_1|+1}(\delta|p,p_{I_1})\psi^{(g_2)}_{|I_2|+1}(\delta|p',p_{I_2})\biggr)
			\int\limits_q^p\int\limits_{q'}^{p'}\Delta B
		\end{align}
		is equivalent to the standard deformation formula in the context of the original CEO topological recursion. In order to see this we note that for $\delta\not=0$ the generalized topological recursion coincides with the original CEO topological recursion (for a local function $x(\delta)$ obtained as some primitive of $dx(\delta)$ on $\sqcup_{q\in \cP}D_q$). Moreover, the differentials $\psi^{(g)}_n(\delta)$ have no residues at the points $q\in \cP(\delta)$, $\delta\not=0$, which fixes the discrepancy with the choice of the lower integration limits for $\Delta B$.
		\item Finally, note that since the generalized topological recursion is compatible with the limits~\cite[Theorem 5.3]{alexandrov2024degenerateirregulartopologicalrecursion}, we see that in the limit $\delta\to 0$ we obtain the statement of the lemma in the context of generalized topological recursion.
	\end{itemize}
\end{proof}

Now we can use this deformation formula to prove Lemma~\ref{lem:Independence-of-B}.

\begin{proof}[Proof of Lemma~\ref{lem:Independence-of-B}]
	
	Note that any two different choices of $B$ can be connected by a path of admissible choices. Thus it is sufficient to prove that the infinitesimal deformation of $B(\epsilon) = B+\epsilon \Delta B+ O(\epsilon^2)$ --- according to Definition~\ref{def:blobbedTR} this deformation nontrivially affects $\psi_n^{B(\epsilon)} = \psi_n^B + \epsilon \Delta \psi_n^B + O(\epsilon^2)$ in~\eqref{eq:89} and  the shifts by $B(\epsilon) = B+\epsilon \Delta B+ O(\epsilon^2)$ in~\eqref{eq:90} --- leaves the $\omega_n(\epsilon)$ infinitesimally intact, that is,  $\omega_n(\epsilon) = \omega_n + O(\epsilon^2)$.
	
	Thus, we have to compute the coefficient of $\epsilon^1$ in the following expression:
	\begin{align}
		\{\omega_n(\epsilon)-\delta_{n,2} B(\epsilon)\}
		& \coloneqq \{\psi^{B(\epsilon)}_n-\delta_{n,2}B(\epsilon)\} *  \{\phi_n -\delta_{n,2} B(\epsilon)\},
	\end{align}
	where $[\epsilon^1]\psi^{B(\epsilon)}_n$ is given by Equation~\eqref{eq:Deformation-B}.
	
	According to Definition~\ref{def:convolution-system}, there are two types of terms that contribute to $[\epsilon^1]\omega_n(\epsilon)-\delta_{n,2} \Delta B$ given as a sum over graphs. We basically have to repeat this definition with the either one of the $\psi$-vertices or one of the $\phi$-vertices being distinguished.
	In the case we have a distinguished $\psi$-vertex, it is decorated by $\Delta\tilde\psi_n = [\epsilon^1]\psi^{B(\epsilon)}_n- \delta_{n,2} \Delta B$. In the case we have a distinguished $\phi$-vertex of index $2$, it is decorated by $-\Delta B$. Pictorially, we represent $[\epsilon^1]\omega_n(\epsilon)-\delta_{n,2} \Delta B$ as the sum over graphs of the following two types:
	\begin{align} \label{eq:TwoTypesOfGraphs}
		\vcenter{
			\xymatrix@C=10pt@R=0pt{
				\txt{\tiny 1} \ar@{-}[dr] &   & \rule{30pt}{0pt} & &  &
				\\
				\txt{\tiny 2} \ar@{-}[r] &  *+[o][F-]{} \ar@{}[u]|(.8){\tilde\psi} \ar@{-}@/^5pt/[rrd]|{*}  & &   &  &
				\\
				\vdots \ar@{-}[ur] & & &  *+[o][F*]{}
				\ar@{}[lu]_(.1){\tilde\phi} & &
				\\
				\vdots \ar@{-}[r] & *++[o][F=]{} \ar@{}[u]|(.6){\Delta\tilde \psi} \ar@{-}@/^5pt/[rru]|{*} \ar@{-}@/_5pt/[rru]|{*}  \ar@{-}@/_5pt/[drr]|{*}& &   & &
				\\
				&  & & *+[o][F*]{}  \ar@{}[lu]_(.1){\tilde\phi}  & \vdots \ar@{-}[l] &
				\\
				\vdots \ar@{-}[r]& *+[o][F-]{} \ar@{}[u]|(.4){\tilde\psi} \ar@{-}@/_5pt/[urr]|{*} & & &\txt{\tiny n} \ar@{-}[lu] &
				\\
				\vdots \ar@{-}[ur]
			}
		}
		\qquad \text{and} \qquad
		\vcenter{
			\xymatrix@C=10pt@R=0pt{
				\txt{\tiny 1} \ar@{-}[dr] &   & \rule{30pt}{0pt} & &  &
				\\
				\txt{\tiny 2} \ar@{-}[r] &  *+[o][F-]{} \ar@{}[u]|(.8){\tilde\psi} \ar@{-}@/^5pt/[rrd]|{*}  & &   &  &
				\\
				\vdots \ar@{-}[ur] & & &   *++[o][F-]{} \save *+[o][F*]{} \restore 
				\ar@{}[lu]_(.3){-\Delta B} & &
				\\
				& *+[o][F-]{} \ar@{}[u]|(.6){\tilde\psi} \ar@{-}[rru]|{*} \ar@{-}@/_5pt/[drr]|{*}& &   & &
				\\
				&  & & *+[o][F*]{}  \ar@{}[lu]_(.1){\tilde\phi}  & \vdots \ar@{-}[l] &
				\\
				\vdots \ar@{-}[r]& *+[o][F-]{} \ar@{}[u]|(.4){\tilde\psi} \ar@{-}@/_5pt/[urr]|{*} & & &\txt{\tiny n} \ar@{-}[lu] &
				\\
				\vdots \ar@{-}[ur]
			}
		}
	\end{align}
	Note that the distinguished $\phi$-vertex in the graphs of the second type might also be connected to a leaf, or, in one exceptional graph, it might be connected to two leaves:
	\begin{align} \label{eq:ExceptionalGraph}
		\vcenter{
			\xymatrix@C=10pt@R=0pt{
				& \txt{\tiny 1} \ar@{-}[dl]
				\\
				*++[o][F-]{} \save *+[o][F*]{} \restore 
				\ar@{}[u]^(0.8){-\Delta B} &
				\\
				& \txt{\tiny 2} \ar@{-}[ul]
			}
		}
	\end{align}
	The weight of this exceptional graph of the second type is equal to $-\Delta B(p_1,p_2)$.

	In terms of these graphs (recall the convention that we use $\tilde \psi_m$ and $\tilde \phi_m$ in the decorations, that is, we subtract $B$), Equation~\eqref{eq:Deformation-B} can equivalently be stated as
	\begin{align} \label{eq:deformation-B-graphs}
		\vcenter{
			\xymatrix@C=10pt@R=0pt{
				\txt{\tiny 1} \ar@{-}[dr] &
				\\
				\vdots \ar@{-}[r] &  *++[o][F=]{} \ar@{}[u]|(0.7){\Delta\tilde \psi}
				\\
				\txt{\tiny n} \ar@{-}[ru] &
			}
		}
		& = -\sum_{i=1}^n
		\vcenter{
			\xymatrix@C=10pt@R=5pt{
				\txt{\tiny 1} \ar@{-}[dr] &  & &  \rule{30pt}{0pt}
				\\
				\hat{\vcenter{\hbox{\txt{\tiny\textit{i}}}}} \ \  \vdots \ar@{-}[r] &  *+[o][F-]{}\ar@{}[u]|(.5){\tilde\psi} \ar@{-}[rr]|{*} & & *++[o][F-]{} \save *+[o][F*]{} \restore 
				\ar@{}[lu]_(.3){-\Delta B} & {\vcenter{\hbox{\txt{\tiny\textit{i}}}}} \ar@{-}[l]
				\\
				\txt{\tiny n} \ar@{-}[ru] & & &
			}
		}
		- \vcenter{
			\xymatrix@C=10pt@R=0pt{
				\txt{\tiny 1} \ar@{-}[dr] &  & \rule{25pt}{0pt} &
				\\
				\vdots \ar@{-}[r] &  *+[o][F-]{} \ar@{}[u]|(0.7){\tilde\psi} \ar@{-}@/_5pt/[rr]|{*} \ar@{-}@/^5pt/[rr]|{*} & & *++[o][F-]{} \save *+[o][F*]{} \restore 
				\ar@{}[lu]_(.3){-\Delta B}
				\\
				\txt{\tiny n} \ar@{-}[ru] &
			}
		}
		\\ \notag & \quad
		-\sum_{I\sqcup J}
		\vcenter{
			\xymatrix@C=10pt@R=0pt{
				&  & &  \rule{30pt}{0pt}
				\\
				I \big \{ \vdots \ \ &  *+[o][F-]{}\ar@{}[u]|(.6){\tilde\psi} \ar@{-}[rrd]|{*} \ar@{-}[l]+<5pt,7pt> \ar@{-}[l]+<5pt,-7pt> & &
				\\
				& & & *++[o][F-]{} \save *+[o][F*]{} \restore 
				\ar@{}[lu]_(.3){-\Delta B}
				\\
				J \big \{ \vdots \ \ &  *+[o][F-]{}\ar@{}[u]|(.6){\tilde\psi}  \ar@{-}[rru]|{*} \ar@{-}[l]+<5pt,7pt> \ar@{-}[l]+<5pt,-7pt> & &
			}
		}
	\end{align}
	Note that the order of the automorphism group of the second summand matches the coefficient $\frac 12$ in the second summand of Equation~\eqref{eq:Deformation-B}.
	
	Now we substitute Equation~\eqref{eq:deformation-B-graphs} instead of the distinguished $\psi$-vertex in each of the graphs of the first type in the expression for $[\epsilon^1]\omega_n(\epsilon)-\delta_{n,2} \Delta B$. Here we face the following problem: for some graphs the new distinguished $\phi$-vertex (decorated by $-\Delta B$) is further connected to a $\phi$-vertex. This type of graphs has never occurred before. Note, however, that since $\Delta B$ and $\tilde \phi_n$ are regular near the points $q_i\in\cP$, the convolution $*$ applied to $-\Delta B$ and $\tilde \phi_n$ vanishes. In other words, the only graphs that contribute after this substitution are the ones where the two edges incident to the new distinguished $\phi$-vertex are either connected to $\psi$-vertices or to the leaves. Thus by this substitution we obtain graphs of the second type.
	
	Now each graph of the second type except for the exceptional one~\eqref{eq:ExceptionalGraph} enters the formula for $[\epsilon^1]\omega_n(\epsilon)-\delta_{n,2} \Delta B$ exactly twice. First, as the original contribution of the second type in~\eqref{eq:TwoTypesOfGraphs} and second, as a summand in the substitution of~\eqref{eq:deformation-B-graphs} in the graphs of the first type in~\eqref{eq:TwoTypesOfGraphs}, with the opposite signs. Thus the only graph that give a non-trivial contribution to the formula for $[\epsilon^1]\omega_n(\epsilon)-\delta_{n,2} \Delta B$ is the exceptional graph~\eqref{eq:ExceptionalGraph}, and we have
	\begin{align}
		[\epsilon^1]\omega_n(\epsilon)-\delta_{n,2} \Delta B = -\delta_{n,2} \Delta B.
	\end{align}
	We see that the infinitesimal deformation of $\{\omega_n\}$ with respect to an arbitrary infinitesimal deformation of $B$ is trivial, which, as we discussed above, proves that the blobbed differentials do not depend on the choice $B$.
\end{proof}

\subsection{Compatibility with non-perturbative differentials}

\begin{proposition}
Assume that $\Sigma$ is compact, $B$ is the standard Bergman kernel fixed by the requirement of vanishing $\mathfrak{A}$-periods, and $\{\phi_n\}$ are the KP integrable differentials of Krichever construction, see Section~\ref{sec:Kr-constr}. In this case the blobbed differentials associated with the spectral curve data $(\Sigma,dx,dy,\cP,\{\phi_n\}_{n\geq 1})$ coincide with nonperturbative differentials $\omega_n^{\np}$ of~\cite{ABDKSnp}.
\end{proposition}

\begin{proof} The nonperturbative differentials are defined in~\cite{ABDKSnp} by a formula with summation over graphs similar to that one of Lemma~\ref{lem:Gamma-accent} but with a different expression for the contribution of the graphs. An equivalence of the two expressions for the contributions of the graphs follows essentially from the Riemann bilinear relations. In more details, let $\overset \times \Sigma$ be the $4g$-gon obtained as a cut of $\Sigma$ along the $\mathfrak{A}$- and $\mathfrak{B}$-cycles.  Let $\{\psi_n\}$ be the differentials associated by generalized topological recursion to the input data $(\Sigma,dx,dy,\cP,B)$, where $B$ is the Bergman kernel normalized on $\mathfrak{A}$-cycles. By assumption, $\tilde\phi_n(p_{\set n}) =  \prod_{i=1}^n(\eta(p_i)\partial_w)\log\theta(w)$, and since these differentials are holomorphic on $\Sigma$, we can rewrite the sum of the residues in the definition of the convolution $*$ as an integral over $\partial\overset \times \Sigma$. More precisely, we have:
\begin{align}
	&
	\sum_{q\in \cP}\res\limits_{p=q} \tilde\psi_{n+1}(p,p_{\set{n}}) \int_q^p \eta(\cdot) \partial_w \big(\prod_{i=1}^{n'}\eta(p'_i) \partial_w\big) \log\theta(w)
	\\ \notag
	& = \frac{1}{2\pi\ii}\oint_{p\in\partial\overset\times\Sigma} \tilde\psi_{n+1}(p,p_{\set{n}}) \int_o^p \eta(\cdot) \partial_w \big(\prod_{i=1}^{n'}\eta(p'_i) \partial_w\big) \log\theta(w),
\end{align}
where $o\in\Sigma$ is some fixed point not on the boundary of $\overset \times \Sigma$ (recall that the residues of $\tilde\psi_{n+1}$ at the points of $\cP$ vanish).

In order to shorten the notation in what follows, we omit the index $n+1$ of $\tilde \psi$, the arguments $p_{\set{n}}$, and the form $\big(\prod_{i=1}^{n'}\eta(p'_i) \partial_w\big) \log\theta(w)$ (to which we apply $\eta(\cdot) \partial_w$)in the subsequent computation. Recall also that since the Bergman kernel is normalized on $\mathfrak{A}$-cycles, we have $\int_{\mathfrak{A_j}}\tilde \psi = 0$. Thus, we have:
\begin{align}
	\frac{1}{2\pi\ii} \oint_{p\in\partial\overset\times\Sigma} \tilde\psi(p) \int_o^p \eta \partial_w & = \frac{1}{2\pi\ii} \oint_{p\in\partial\overset\times\Sigma} \tilde\psi(p) \int_o^p \sum_{i=1}^g \eta_i(\cdot) \partial_{w_i}
	\\ \notag
	& = - \frac{1}{2\pi\ii} \sum_{j=1}^g \oint_{ \mathcal{A}_j} \tilde\psi \oint_{\mathfrak{B}_j} \eta(\cdot) \partial_w + \frac{1}{2\pi\ii} \sum_{j=1}^g \oint_{ \mathcal{B}_j} \tilde\psi \oint_{\mathfrak{A}_j} \eta(\cdot) \partial_w
	\\ \notag
	& =  \frac{1}{2\pi\ii} \oint_{ \mathcal{B}} \tilde\psi \partial_w.
\end{align}
The latter formula was used instead of the convolution $*$ on the edges of graphs in the definition of $\omega_n^{\np}$ in~\cite{ABDKSnp}. Since further the structure of graphs in~\cite[Section 4.1]{ABDKSnp} coincides with the version of the convolution formula presented in the proof of Proposition~\ref{prop:def-omegas} (in terms of graphs $G'_n$, with $\{\phi_n\}$ and $\{\psi_n\}$ interchanged), we identify the blobbed differentials associated with the spectral curve data $(\Sigma,dx,dy,\cP,\{\phi_n\}_{n\geq 1})$ with the nonperturbative differentials $\omega_n^{\np}$ of~\cite{ABDKSnp}.
\end{proof}

\subsection{KP integrability}

\begin{theorem}\label{th:KP-blobbed}
Assume that the system of blobs $\{\phi_n\}$ in the initial data of blobbed topological recursion is KP integrable. Then the system of differentials of the corresponding blobbed topological recursion is also KP integrable.
\end{theorem}

This theorem generalizes and unifies the KP integrability of the non-perturbative tau function conjectured in~\cite{BorEyn-AllOrderConjecture} and proved in~\cite{ABDKSnp}. It also provides the right context to the result of~\cite{alexandrov2024topologicalrecursionrationalspectral} claiming that the differentials of topological recursion on a global spectral curve is KP integrable if the spectral curve is rational. Indeed, by Theorem~\ref{th:KP-blobbed}, the KP integrability of the system of TR differentials follows from the KP integrability of the system of differentials~$\phi_n=\delta_{n,2}B$ considered as blobs, and by Proposition~\ref{prop:trivialKP}, such system is KP integrable if the spectral curve is rational (and it can be upgraded to an ``only if'' statement, see also~\cite{ABDKS3}). Notice, however, that the proof of Theorem~\ref{th:KP-blobbed} given below rely on the mentioned result of~\cite{alexandrov2024topologicalrecursionrationalspectral}.

\begin{proof} By Lemma~\ref{lem:Independence-of-B}, we have a freedom in a choice of $B$. Let us choose it in a special form
\begin{equation}
B(p_1,p_2)=\frac{dz(p_1)dz(p_2)}{(z(p_1)-z(p_2))^2}
\end{equation}
for some meromorphic function~$z$. We may assume that $z$ takes finite pairwise distinct values at the points of $\cP$, and $dz\ne0$ at these points. Along with the pole on the diagonal, the differential $B$ may have some additional poles, namely, at the pairs of distinct points $(p_1,p_2)$ with equal values of $z$. We can exclude these extra poles by
replacing the curve $\Sigma$ by its open subset $U$ containing~$\cP$. It is even sufficient to assume that $U$ is a union of small disks centered at the points of~$\cP$. The function~$z$ provides an embedding of $U$ to~$\C P^1$, and the differential~$B$ coincides with the standard Bergman kernel on~$\C P^1$ under this embedding. Thus, replacing $\Sigma$ by $U$, we may assume that the spectral curve in the initial data of the topological recursion defining the differentials $\psi^B_n$ is rational. The differentials $dx$ and $dy$ are defined locally in a neighborhood of~$\cP$ but the differentials $\psi^B_n$ still extend as global meromorphic differentials on~$(\C P^1)^n$.

By~\cite[Theorem 1.1]{alexandrov2024topologicalrecursionrationalspectral} and~\cite[Theorem 6.4]{alexandrov2024degenerateirregulartopologicalrecursion}, the differentials~$\psi^B_n$ are KP integrable. By Theorem~\ref{th:KP-conv}, we conclude that the differentials $\omega_n$ of blobbed topological recursion are also KP integrable, at least for the restriction to $U\subset\Sigma$. But since the KP integrability for a system of differentials is a global property, it holds for the whole spectral curve~$\Sigma$.
\end{proof}

\subsection{Recursions}

The differentials of blobbed topological recursion are defined in Section~\ref{sec:TR-blobbed} by a combinatorial formula with summation over graphs. In this section, we present a number of possible recursions for these differentials simplifying their computation in practice.

\begin{proposition}
The differentials of blobbed topological recursion satisfy generalized loop equations, the same as for the differentials of GenTR:
\begin{equation}
\sum_{r\ge0}\bigl(-d\tfrac1{dy}\bigr)^r[u^r]e^{u(\cS(u\hbar\partial_x)-1)y}
\restr{q^\pm}{ \mathfrak{g}^{\pm\frac{u\hbar}{2}}_{\partial_x}p}\cW_n(q^+,q^-;p_{\set n})
\end{equation}
is holomorphic in~$p$ at any $q\in\cP$, where $\cW_n$ is the extended differential associated with the system of differentials $\{\omega_n\}$.
\end{proposition}

Let $G_{\ell,m}\subset G_{\ell+m}$ be the set of graphs of Definition~\ref{def:convolution-system} with $\ell$ $\psi$-leaves marked by $1,\dots,\ell$ and $m$ $\phi$-leaves marked by $\ell+1,\dots,\ell+m$. Set
\begin{equation}
\omega_{\ell,m}=\sum_{\Gamma\in G_{\ell,m}}\frac{\mathsf{w}_\Gamma}{|\Aut(\Gamma)|}
\end{equation}
and denote by $\omega_{\ell,m}^{\langle d\rangle}$ the coefficient of $\hbar^d$ in $\omega_{\ell,m}$. Accordingly, we have
\begin{equation}
\omega_n(p_{\set{n}})=\sum_{I\sqcup J=\set{n}}\omega_{|I|,|J|}(p_I,p_J)+\delta_{n,2}B(p_1,p_2).
\end{equation}

\begin{proposition}
The blobbed TR differentials $\omega_{\ell,m}^{\langle d\rangle}$ with $\ell>0$ satisfy projection property
\begin{equation}
\omega_{\ell,m}^{\langle d\rangle}(p,p_{\set{\ell+m-1}})=\sum_{q\in\cP}\res_{\tilde p=q}\omega_{\ell,m}^{\langle d\rangle}(\tilde p,p_{\set{\ell+m-1}})
\int_q^{\tilde p}B(\cdot,p).
\end{equation}
\end{proposition}

The loop equation define the principal parts of the poles of $\omega^{\langle d\rangle}_{\ell,m}$ at the points of $\cP$ with respect to the first argument, from the known differentials computed in the previous steps of recursion, and the projection property recovers $\omega^{\langle d\rangle}_{\ell,m}$ from the principal parts of its poles. This provides a recursive procedure allowing one to compute all $\omega_{\ell,m}$ with $\ell>0$. However, the computation of $\omega_{0,m}$ is not covered by this recursion.

\medskip
Now, let us split
\begin{equation}
\omega_{\ell,m}=\sum_{k=0}^\infty \omega_{\ell,m,k}
\end{equation}
where $\omega_{\ell,m,k}$ involves the contribution of \emph{graphs with $k$ internal edges} in $G_{\ell,m}$. Then, for $k>0$ we have
\begin{multline}
k\,\omega_{\ell,m,k}(p_{\set{\ell+m}})=*\biggl(\omega_{\ell+1,m+1,k-1}(p,p_{\set{\ell}},\bar p,p_{\set{\ell+m}\setminus\set{\ell}})
\biggr.\\+\biggl.
\sum_{\substack{I_1\sqcup I_2=\set{\ell}\\J_1\sqcup J_2=\set{\ell+m}\setminus\set{\ell}\\k_1+k_2=k-1}}
\omega_{|I_1|+1,|J_1|,k_1}(p,p_{I_1},p_{J_1})\omega_{|I_2|,|J_2|+1,k_2}(p_{I_2},\bar p,p_{J_2})
\biggr)
\end{multline}
where
\begin{equation}
*\omega(p,\bar p)=\sum_{q\in\cP}\res_{p=q}\int_q^p\omega(p,\cdot).
\end{equation}
These relations provide another sort of recursion allowing one to compute $\omega_{\ell,m,k}$ by induction in~$k$. The initial differentials for this recursion are
\begin{equation}
\omega_{\ell,0,0}=\tilde \psi^B_\ell,
\qquad \omega_{0,m,0}=\tilde \phi_{m},
\qquad \omega_{\ell,m,0}=0, \quad \ell>0,m>0.
\end{equation}

\printbibliography

\end{document}